\documentclass[a4paper,UKenglish]{lipics-v2018} 

\usepackage{amsmath,amsfonts,latexsym}
\usepackage{graphicx}
\usepackage{float}
\usepackage{enumerate}
\usepackage{multirow}
\usepackage{enumerate}
\usepackage{xcolor}
\usepackage{mathtools}
\usepackage{amssymb}
\usepackage{braket}
\usepackage{psfrag, epstopdf}
\usepackage{cite}
\usepackage{algorithm2e}
\usepackage{algorithmic}
\usepackage{caption}
\usepackage{pstricks, pst-3dplot}

\usepackage{tikz}
\usepackage{tikz-3dplot}
\definecolor{Darkgreen}{rgb}{0,0.4,0}
\definecolor{Darkblue}{rgb}{0,0,0.4}
\definecolor{ggcol}{rgb}{0.9,1,0.9}
\definecolor{bbcol}{rgb}{0.9,0.9,1}
\definecolor{rrcol}{rgb}{1,0.9,0.9}
\usepackage{hyperref}
\usepackage{pgfplots}
\usetikzlibrary{calc,3d,intersections, positioning,intersections,shapes}

\newcommand{\Z}{{\mathbb{Z}}}  
\newcommand{\R}{{\mathbb{R}}}  
\newcommand{\Sph}{{\mathcal{S}}}  

\newcommand{\dist}[1]{\zeta(#1)}

\newlength\myindent
\newcommand\bindent{%
  \begingroup
  \setlength{\itemindent}{\myindent}
  \addtolength{\algorithmicindent}{\myindent}
}
\newcommand\eindent{\endgroup}

\title{Towards Uniform Online Spherical Tessellations}

\titlerunning{Towards Uniform Online Spherical Tessellations}

\author{Paul C. Bell}{Department of Computer Science, James Parsons Building, Byrom Street, Liverpool John Moores University, Liverpool, L3-3AF, UK}{p.c.bell@ljmu.ac.uk}{ https://orcid.org/0000-0003-2620-635X}{}

\author{Igor Potapov}{Department of Computer Science, Ashton Building, Ashton Street, University of Liverpool, Liverpool, L69-3BX, UK}{potapov@liverpool.ac.uk}{}{Research partially funded by the grant EP/R018472/1 ``Application driven Topological Data Analysis''}

\authorrunning{P.\,C. Bell and I.\, Potapov}

\Copyright{Paul C. Bell and Igor Potapov}



\nolinenumbers 
\hideLIPIcs  

\begin{document}

\maketitle

\begin{abstract}
The problem of uniformly placing $N$ points onto a sphere finds applications in many areas. For example, points on the sphere correspond to unit quaternions as well as to the group of rotations SO$(3)$ and the online version of generating uniform rotations (known as ``incremental generation'') plays a crucial role in a large number of  engineering applications ranging from robotics and aeronautics to computer graphics.
An online version of this problem was recently studied with respect to the \emph{gap ratio} as a measure of uniformity. The first online algorithm of Chen et al. was upper-bounded by $5.99$ and later improved to $3.69$, which is  achieved by considering a circumscribed dodecahedron followed by a recursive decomposition of each face. 

In this paper we provide  a more efficient tessellation technique based on the regular icosahedron, which improves the upper-bound for the online version of this problem, decreasing it
  to approximately $2.84$. Moreover, we show that the lower bound for the gap ratio of placing at least three points is $\frac{1+\sqrt{5}}{2} \approx 1.618$ and for at least four points is no less than $1.726$.  
 \end{abstract}

\tdplotsetmaincoords{90}{90}

\section{Introduction}
One of the central problems of classical discrepancy theory is to maximize the uniformity of distributing a set of $n$ points into some metric space \cite{Discrepancy1991,Discrepancy2000}.  For example, this includes questions about arranging points over a unit cube in a $d$-dimensional space, a polyhedral region, a sphere, a torus or even over a hyperbolic plane, etc. Applications of modern day discrepancy theory include those in viral morphology, crystallography, molecular structure, electrostatics, computing quadrature formulas \cite{Saff97}, 3D projection reconstruction of Magnetic Resonance Images (MRI) \cite{Ko14}, investigations of carbon fullerenes (C$_{60}$ and C$_{70}$) \cite{CS91}, complexity theory \cite{tammes}, connections to quaternion rotations \cite{quaternion}, Ramsey theory, problems in numerical integration, financial calculations, computer graphics and computational physics \cite{Discrepancy1999}. The choice of discrepancy measure impacts the type of questions that can be studied. For example, a logarithmic variant of the Coulomb potential allows one to define Smale's $7^{\text{th}}$ problem \cite{smale}; from a list of $18$ important questions for the $21^{\text{st}}$ century. 

In order to measure the discrepancy from uniformity, the \emph{gap ratio} metric was introduced by Teramoto et al.~\cite{Teramoto} for analysing dynamic discrepancy, i.e. the ratio between the maximum  and minimal gap, where the maximum gap is the diameter of the largest empty circle on the sphere surface  and the minimal gap is the minimum pairwise distance on the sphere surface (orthodromic distance). A similar measure of taking the radius of the largest empty circle over  minimum pairwise distance is  
known as the ``mesh ratio''  \cite{BHS12} or ``mesh-separation ratio''  \cite{JComplexity2015,Etayo2019}.
One might consider defining uniformity just by measuring the closest two points, however this does not take into account large undesirable gaps that may be present in the point set. 
Alternatively one may use the standard measure from discrepancy theory; define some fixed geometric shape $R$ and count the number of inserted points that are contained in $R$, whilst moving it all over some space.
This measure has two main disadvantages -- that of computational hardness of calculating the discrepancy at each stage and also that we must decide upon a given shape $R$, each of which may give different results \cite{Teramoto}. 

Generating a set of uniformity distributed points in non-Euclidean space is even more challenging and somewhat counterintuitive.
For example the problem  of generating a point set on the $2$-sphere  which minimizes criteria such as energy functions, discrepancy, dispersion and mutual distances has been extensively studied  in the offline setting \cite{Hardin04, Lub86, Saff97, Sloane97, Sun08, Wag93, YJLM}.
Some motivations and applications of this problem when restricted to the $2$-sphere stretch from the classical \emph{Thompson problem} of determining a configuration of $N$ electrons on the surface of a unit sphere that minimizes the electrostatic potential energy\cite{thomp1904, Rak94}, to search and rescue/exploration problems 
as well as problems related to extremal energy, crystallography and computational chemistry \cite{Saff97}. 
In the original offline version of the problem of distributing points over some space, the number of points is predetermined and the goal is to distribute all points as uniformly as possible at the end of the process. 

In contrast, the \emph{online} optimisation problem requires that the points should be dynamically inserted one at a time on the surface of a sphere without knowing the total number of points in advance. 
Thus, the objective is to distribute the points as uniformly as possible at every instance of inserting a point. Note that in this case
once a point has been placed \emph{it cannot be later moved}.  
The points on the sphere correspond to unit quaternions and the group of rotations SO$(3)$ \cite{Kirk,quaternion}, which have a large number of engineering applications in the online setting (also known as ``incremental generation'') and plays a crucial role in applications ranging from robotics and aeronautics to computer graphics \cite{YJLM}.

The online variant of distributing points in a given space  has already been studied, e.g. 
for inserting integral points on a line \cite{IPL2008} or on a grid \cite{Zhang11},
inserting real points over a unit cube \cite{Teramoto} and also recently as a more complex version of inserting real points on the surface of a sphere \cite{walcom17,YJLM,ZC18}. 
A good strategy for online distribution of points on the plane has been found in \cite{IPL2007,Teramoto}
based on the Voronoi insertion, where the gap ratio is proved to be at most $2$. For insertion on a two-dimensional grid, algorithms with a maximal gap ratio $2\sqrt{2} \approx 2.828$  were shown in  \cite{Zhang11}. The same authors showed that the lower bound for the maximal gap ratio is $2.5$ in this context. The other important direction was to solve the problem for the one-dimensional line and an insertion strategy with a uniformity of $2$ has been found in \cite{IPL2008}. An approach of using \emph{generalised spiral points} was discussed in \cite{Rak94, Saff97}, which performs well for minimizing extremal energy, but this approach is strictly offline (number of points $N$ known in advance).

Recently the problem of {\sl online distribution of points on the $2$-sphere} has been proposed in 
\cite{walcom17} where a \emph{two phase point insertion algorithm} with an overall upper bound of $5.99$ was designed. The first phase uses a circumscribed dodecahedron to place the first twenty vertices, achieving a maximal gap ratio $2.618$. After that, each of the twelve pentagonal faces can be recursively divided. 
This procedure is efficient and leads to a gap ratio of no more than $5.99$. With  more complex analysis, the bound was recently decreased to at most $3.69$ in \cite{ZC18}.

One may consider whether such two phase algorithms may perform well for this problem by either modifying the initial shape used in the first phase of the algorithm (such as using initial points derived from other Platonic solids) or else whether the recursive procedure used in phase two to tessellate each regular shape may be improved. We may readily identify an advantage to choosing a Platonic solid for which each face is a triangle (the tetrahedron, octahedron and icosahedron), since in this case at least two procedures for tessellating each triangle immediately spring to mind -- namely to recursively place a new point at the centre of \emph{each edge}, denoted \emph{triangular dissection} (creating four subtriangles, see Fig.~\ref{normaltess}), or else the Delaunay tessellation, which is to place a new point at the \emph{centre of each triangle} (creating three new triangles, see Fig.~\ref{deltess}). 

\begin{figure}[t]
\begin{minipage}[c]{0.27\linewidth}
\begin{center}
          \begin{tikzpicture}[scale=1.5]
\draw [black, fill=ggcol] (0.9428    ,     0) -- (-0.4714  , -0.8165) -- (0.1361 ,  -0.2357) -- cycle;
\draw [black, fill=bbcol] (0.9428   ,      0) -- (0.1361,   -0.2357) -- (0 ,    0) -- cycle;
\draw [black, fill=rrcol] (-0.4714,  -0.8165) -- (0.1361 ,  -0.2357) -- (0 ,    0) -- cycle;
\draw [black, fill=ggcol] (0.9428   ,      0) -- (-0.4714  ,  0.8165) -- (0.1361 ,   0.2357) -- cycle;
\draw [black, fill=bbcol] (-0.4714 ,   0.8165) -- (0.1361 ,   0.2357) -- (0 ,    0) -- cycle;
\draw [black, fill=rrcol] (0.9428   ,     0) -- (0.1361  ,  0.2357) -- (0 ,    0) -- cycle;
\draw [black, fill=ggcol] (-0.4714  ,  0.8165) -- (-0.4714 ,  -0.8165) -- (-0.2722    ,     0) -- cycle;
\draw [black, fill=bbcol] (-0.4714 ,  -0.8165) -- (-0.2722   ,      0) -- (0 ,    0) -- cycle;
\draw [black, fill=rrcol] (-0.4714 ,   0.8165) -- (-0.2722   ,      0) -- (0 ,    0) -- cycle;
          \end{tikzpicture}
        \end{center}
\caption{\footnotesize Delauney triangulation applied twice}\label{deltess}
\end{minipage}
\hspace{0.7cm}
\begin{minipage}[c]{0.27\linewidth}
\begin{center}
          \begin{tikzpicture}[scale=1.5]
\draw [black, fill=ggcol] (0.9428    ,     0) -- (0.5893 ,   0.2041) -- (0.5893,   -0.2041) -- cycle;
\draw [black, fill=ggcol] (0.2357 ,   0.4082) -- (0.5893 ,   0.2041) -- (0.2357 ,        0) -- cycle;
\draw [black, fill=bbcol] (0.5893   , 0.2041) -- (0.2357 ,        0) -- (0.5893  , -0.2041) -- cycle;
\draw [black, fill=ggcol] (0.2357   ,      0) -- (0.5893,   -0.2041) -- (0.2357,   -0.4082) -- cycle;
\draw [black, fill=ggcol] (-0.1179 ,   0.6124) -- (0.2357 ,   0.4082) -- (-0.1179 ,   0.2041) -- cycle;
\draw [black, fill=bbcol] (0.2357 ,   0.4082) -- (-0.1179 ,   0.2041) -- (0.2357,         0) -- cycle; %
\draw [black, fill=ggcol] (-0.1179 ,   0.2041) -- (0.2357   ,      0) -- (-0.1179 ,  -0.2041) -- cycle;
\draw [black, fill=bbcol] (0.2357   ,      0) -- (-0.1179 ,  -0.2041) -- (0.2357 ,  -0.4082) -- cycle;
\draw [black, fill=ggcol] (-0.1179 ,  -0.2041) -- (0.2357 ,  -0.4082) -- (-0.1179,  -0.6124) -- cycle;
\draw [black, fill=ggcol] (-0.4714 ,   0.8165) -- (-0.1179 ,   0.6124) -- (-0.4714  ,  0.4082) -- cycle;
\draw [black, fill=bbcol] (-0.1179 ,   0.6124) -- (-0.4714 ,   0.4082) -- (-0.1179 ,   0.2041) -- cycle;
\draw [black, fill=ggcol] (-0.4714 ,   0.4082) -- (-0.1179  ,  0.2041) -- (-0.4714    ,     0) -- cycle;
\draw [black, fill=bbcol] (-0.1179  ,  0.2041) -- (-0.4714    ,    0) -- (-0.1179 ,  -0.2041) -- cycle;
\draw [black, fill=ggcol] (-0.4714  ,       0) -- (-0.1179,  -0.2041) -- (-0.4714  , -0.4082) -- cycle;
\draw [black, fill=bbcol] (-0.1179 ,  -0.2041) -- (-0.4714 ,  -0.4082) -- (-0.1179 ,  -0.6124) -- cycle;
\draw [black, fill=ggcol] (-0.4714 ,  -0.4082) -- (-0.1179,-0.6124) -- (-0.4714 ,  -0.8165) -- cycle;
          \end{tikzpicture}
        \end{center}
\caption{\footnotesize Recursive triangle dissection (twice)}\label{normaltess}
\end{minipage}
\hspace{0.7cm}
\begin{minipage}[c]{0.27\linewidth}
\begin{center}
          \begin{tikzpicture}[scale=1.5]
\draw [black, fill=ggcol] (0.9428,0) -- (0.7607  ,  0.3981) -- (0.7607  , -0.3981) -- cycle;
\draw [black, fill=ggcol] (0.4082 ,   0.7071) -- (0.7607  ,  0.3981) -- (0.5774     ,    0) -- cycle;
\draw [black, fill=bbcol] (0.7607 ,   0.3981) -- (0.7607 ,  -0.3981) -- (0.5774   ,      0) -- cycle;
\draw [black, fill=ggcol] (0.7607,  -0.3981) -- (0.5774    ,     0) -- (0.4082 ,  -0.7071) -- cycle;
\draw [black, fill=ggcol] (-0.0356 ,   0.8578) -- (0.4082 ,   0.7071) -- (-0.2887   , 0.5000) -- cycle;
\draw [black, fill=bbcol] (0.4082 ,   0.7071) -- (-0.2887  ,  0.5000) -- (0.5774     ,    0) -- cycle;
\draw [black, fill=ggcol] (0.5774    ,     0) -- (-0.2887  ,  0.5000) -- (-0.2887,   -0.5000) -- cycle;
\draw [black, fill=bbcol] (0.5774    ,     0) -- (0.4082 ,  -0.7071) -- (-0.2887 ,  -0.5000) -- cycle;
\draw [black, fill=ggcol] (0.4082 ,  -0.7071) -- (-0.0356 ,  -0.8578) -- (-0.2887,   -0.5000) -- cycle;
\draw [black, fill=ggcol] (-0.4714  ,  0.8165) -- (-0.0356 ,   0.8578) -- (-0.7251  ,  0.4597) -- cycle;
\draw [black, fill=bbcol] (-0.0356  ,  0.8578) -- (-0.7251 ,   0.4597) -- (-0.2887  ,  0.5000) -- cycle;
\draw [black, fill=ggcol] (-0.7251  ,  0.4597) -- (-0.2887  ,  0.5000) -- (-0.8165 ,        0) -- cycle;
\draw [black, fill=bbcol] (-0.2887  ,  0.5000) -- (-0.8165   ,      0) -- (-0.2887 ,  -0.5000) -- cycle;
\draw [black, fill=ggcol] (-0.8165  ,       0) -- (-0.2887 ,  -0.5000) -- (-0.7251  , -0.4597) -- cycle;
\draw [black, fill=bbcol] (-0.2887  , -0.5000) -- (-0.7251 ,  -0.4597) -- (-0.0356  , -0.8578) -- cycle;
\draw [black, fill=ggcol] (-0.7251  , -0.4597) -- (-0.0356  , -0.8578) -- (-0.4714 ,  -0.8165) -- cycle;
          \end{tikzpicture}
        \end{center}
\caption{\footnotesize Deformation of triangles in projection}\label{skewedtess}
\end{minipage}
\end{figure}

It can be readily seen that the Delaunay tessellation of each spherical triangle rapidly gives a poor gap ratio, since points start to become dense around the centre of edges of the initial tessellation. The second recursive tessellation strategy (Fig.~\ref{normaltess}), was conjectured to give a poor ratio in \cite{walcom17}. 
This intuition seems reasonable, since as we recursively decompose each spherical triangle by this strategy the gap ratio increases as for such a triangle this decomposition deforms with each recursive step, as can be seen in Fig.~\ref{skewedtess}. It can also be seen that the gap ratio at each level of the triangular dissection increases (see Lemma~\ref{mainlem}). Nevertheless, we show in this paper that as long as the initial tessellation (stage 1) does not create too `large' spherical triangles (with high curvature), then the gap ratio of stage 2 has an upper limit, and performs much better than the tessellation of the regular dodecahedron proposed in \cite{walcom17} and \cite{ZC18}.

In this paper, we provide a new algorithm and utilise  a circumscribed regular icosahedron and the recursive triangular dissection procedure to reduce the bound of $3.69$ derived in \cite{ZC18} to 
$ \frac{\pi}{\arccos{\left( 1/\sqrt{5}\right)}} \approx 2.8376$.
Apart from a
better upper bound, an advantage of our triangular tessellation procedure is its generalisability and more efficient tessellation as we only need to compute the spherical median between two locally introduced points at every step. 

Another natural point insertion algorithm to consider is a greedy algorithm, where points are iteratively added to the centre of the largest empty circle. However, the decomposition of a 2-sphere according to the greedy approach leads to complex non-regular local structures and it soon becomes intractable to determine the next point to place \cite{walcom17, ZC18} (such points can in general even be difficult to describe in a computationally efficient way). 

A preliminary version of this paper appeared in \cite{BP18}. The present paper provides full details and diagrams for proofs, improves the lower bound of $1.618$ to $1.726$ in Theorem~\ref{lowerboundthm2}, provides some experimental results and shows in Section~\ref{counterSec} a counterexample of a claimed lower bound of $1.78$ from \cite{ZC18}.

\section{Notation}\label{sec-not}

\subsection{Spherical trigonometry}\label{spher-trig}

Given a set $P$, we denote by $2^P$ the power set of $P$ (the set of all subsets). Let $\Sph$ denote the $3$-dimensional unit sphere (the $2$-sphere). We will deal almost exclusively with unit spheres, since for our purposes the gap ratio (introduced formally later) is not affected by the spherical radius. Let $u_1, u_2, u_3 \in \R^3$ be three unit length vectors, then $T = \langle u_1, u_2, u_3\rangle$ denotes the spherical triangle on $\Sph$ with vertices $u_1, u_2$ and $u_3$. Given some set of points $\{u_1, u_2, u_3\} \cup \{v_j | 1 \leq j \leq k\}$, a spherical triangle $T = \langle u_1, u_2, u_3\rangle$ is called \emph{minimal} over that set of points if no $v_j$ for $1 \leq j \leq k$ lies on the interior or boundary of $T$. As an example, in Fig.~\ref{sigmafig}, triangle $\langle u_1, u_{113}, u_{112}\rangle$ is minimal, but $\langle u_1, u_{13}, u_{12}\rangle$ is not, since points $u_{113}$ and $u_{112}$ lie on the boundary of that triangle. 
 
The edges of a spherical triangle are arcs of great circles. A great circle is the intersection of $\Sph$ with a central plane, i.e one which goes through the centre of $\Sph$. We denote the length of a path connecting two points $u_1, u_2$ on the unit sphere by $\zeta(u_1, u_2)$ (the spherical length). 

Given a non-degenerate spherical triangle (i.e. one with positive area, defined later) with two edges $e_1$ and $e_2$ which intersect at a point $P$, then we say that the angle of $P$ is the angle of $P$ measured when projected to the plane tangent at $P$. We constrain all spherical triangles to have edge lengths strictly between $0$ and $\pi$, which avoids issues with \emph{antipodal triangles}. Two points on the unit sphere are called antipodal if the angle between them is $\pi$ (i.e. they lie opposite to each other on the unit sphere) and an antipodal triangle contains two antipodal points. Several results in spherical trigonometry (and in this paper) are derived by projections of points/edges to planes tangent to a point on the sphere; in all such cases the projection is from the centre of the sphere.

The following results are all standard from spherical trigonometry, see \cite{tod1886} for proofs and further details. The length of an arc belonging to a great circle corresponds with the angle of the arc, see Fig.~\ref{figuretansec}. Furthermore, given an arc between two points $u_1$ and $u_2$ on $\Sph$, the length of the line connecting $u_1$ and the projection of $u_2$ to the plane tangent to $u_1$ is given by $\tan(\zeta(u_1, u_2))$, see Fig.~\ref{figuretansec}.

\begin{figure}[ht]
\begin{minipage}[c]{0.45\linewidth}
\begin{center}
          \begin{tikzpicture}[scale=0.8]
\draw [bbcol, fill=bbcol] (0,0) -- (1.678, 2) -- (0, 2) -- cycle;
\draw [ggcol, fill=ggcol] (0,0) -- ++(50:2) -- (0, 2) -- cycle;

\draw[gray,->] (-2.5,0) -- (2.5,0);
\draw[gray,->] (0,-2.5) -- (0,2.5);

\draw [thin, gray] (0,0) circle (2cm);
\draw [thick, blue, fill=ggcol] (0,0) ++(90:2) arc (90:50:2);
\draw [thick, blue] (0,0) ++(90:0.4) arc (90:50:0.4);
\draw (0,0) -- (1.678, 2);
\draw (0,2) -- (1.678, 2);

\node[Darkgreen] at (0.8,2.3) {$\tan(\alpha)$};
\node[Darkgreen] at (1.3,0.7) {$\sec(\alpha)$};
\node[] at (-0.3,1) {$1$};
\node[] at (0.2, 0.6) {$\alpha$};
\node[] at (0.6, 1.6) {$\alpha$};
\node[] at (-0.3, 2.2) {$u_1$};
\node[] at (1.65, 1.56) {$u_2$};
\node[] at (1.578, -2.2) {$\alpha = \zeta(u_1, u_2)$};

          \end{tikzpicture}
        \end{center}
\caption{\footnotesize Angular calculations in the plane intersecting the great circle containing $(u_1, u_2)$}\label{figuretansec}
\end{minipage}
\hspace{1.5cm}
\begin{minipage}[c]{0.45\linewidth}

\begin{center}\label{ffig}
          \begin{tikzpicture}[scale=1.3]
            \pgfmathsetmacro\R{1} 
            \fill[ball color=white!10, opacity=0.1] (0,0,0) circle (\R); 
            \tdplotsetmaincoords{65}{120}
            \begin{scope}[tdplot_main_coords, shift={(0,0)}]
              \coordinate (O) at (0,0,0);


\tdplotdefinepoints(0,0,0)(1, 0, 0)(0,1,0)
\tdplotdrawpolytopearc[thick]{\R}{anchor=north}{\scriptsize $u_{12}$}

\tdplotdefinepoints(0,0,0)(0,1,0)(0, 0, 1)
\tdplotdrawpolytopearc[thick]{\R}{anchor=west}{\scriptsize $u_{23}$}

\tdplotdefinepoints(0,0,0)(1, 0, 0)(0,0,1)
\tdplotdrawpolytopearc[thick]{\R}{anchor=east}{\scriptsize $u_{13}$}

\tdplotdefinepoints(0,0,0)(0.7071, 0.7071, 0)(0.7071,0,  0.7071)
\tdplotdrawpolytopearc[blue]{\R}{anchor=east}{}

\tdplotdefinepoints(0,0,0)(0.7071, 0.7071, 0)(0, 0.7071, 0.7071)
\tdplotdrawpolytopearc[blue]{\R}{anchor=east}{}

\tdplotdefinepoints(0,0,0)(0.7071,0,  0.7071)(0, 0.7071, 0.7071)
\tdplotdrawpolytopearc[blue]{\R}{anchor=east}{}


\tdplotdefinepoints(0,0,0)(0.9239   ,      0   , 0.3827)(0.8165  ,  0.4082  ,  0.4082)
\tdplotdrawpolytopearc[red]{\R}{anchor=east}{}
\tdplotdefinepoints(0,0,0)(0.9239 ,   0.3827    ,     0)(0.8165  ,  0.4082  ,  0.4082)
\tdplotdrawpolytopearc[red]{\R}{anchor=east}{}
\tdplotdefinepoints(0,0,0)(0.8539 ,   0.3827    ,     0)(0.9239   ,      0   , 0.3827)
\tdplotdrawpolytopearc[red]{\R}{anchor=east}{}

\tdplotdefinepoints(0,0,0)(0.3827  ,       0  ,  0.9239)(0  ,  0.3827  ,  0.9239)
\tdplotdrawpolytopearc[red]{\R}{anchor=east}{}
\tdplotdefinepoints(0,0,0)(0.3827  ,       0  ,  0.9239)(0.4082  ,  0.4082 ,   0.8165)
\tdplotdrawpolytopearc[red]{\R}{anchor=east}{}
\tdplotdefinepoints(0,0,0)(0  ,  0.3827  ,  0.9239)(0.4082  ,  0.4082 ,   0.8165)
\tdplotdrawpolytopearc[red]{\R}{anchor=east}{}

\tdplotdefinepoints(0,0,0)(0,    0.9239 ,   0.3827)(0.3827  ,  0.9239    ,     0)
\tdplotdrawpolytopearc[red]{\R}{anchor=east}{}
\tdplotdefinepoints(0,0,0)(0,    0.9239 ,   0.3827)(0.4082 ,   0.8165  ,  0.4082)
\tdplotdrawpolytopearc[red]{\R}{anchor=east}{}
\tdplotdefinepoints(0,0,0)(0.3827  ,  0.9239    ,     0)(0.4082 ,   0.8165  ,  0.4082)
\tdplotdrawpolytopearc[red]{\R}{anchor=east}{}

\tdplotdefinepoints(0,0,0)(0.8165 ,   0.4082  ,  0.4082)(0.4082 ,   0.8165  ,  0.4082)
\tdplotdrawpolytopearc[red]{\R}{anchor=east}{}
\tdplotdefinepoints(0,0,0)(0.8165 ,   0.4082  ,  0.4082)(0.4082  ,  0.4082  ,  0.8165)
\tdplotdrawpolytopearc[red]{\R}{anchor=east}{}
\tdplotdefinepoints(0,0,0)(0.4082 ,   0.8165  ,  0.4082)(0.4082  ,  0.4082  ,  0.8165)
\tdplotdrawpolytopearc[red]{\R}{anchor=east}{}

\draw[] (1,-0.1,0) node[anchor=north] {\scriptsize $u_1$};
\draw[] (-0.1,1.1,0) node[anchor=north] {\scriptsize $u_2$};
\draw[] (0,0,1) node[anchor=east] {\scriptsize $u_3$};

\draw[] (0.9239   ,      0   , 0.3827) node[anchor=east] {\scriptsize $u_{113}$};
\draw[] (0.9165  ,  0.3582  ,  0.4082) node[anchor=south west] {\scriptsize $u_{1213}$};
\draw[] (0.9239 ,   0.3727    ,     0) node[anchor=north] {\scriptsize $u_{112}$};

\draw[] (0.3827  ,       0  ,  0.9239) node[anchor=east] {\scriptsize $u_{133}$};
\draw[] (0.1  ,  0.3827  ,  0.9239) node[anchor=south west] {\scriptsize $u_{323}$};
\draw[] (0.4082  ,  0.4082 ,   0.8165) node[anchor=north] {\scriptsize $u_{1323}$};

\draw[] (0,    0.9239 ,   0.3827) node[anchor=west] {\scriptsize $u_{232}$};
\draw[] (0.3827  ,  0.9239    ,     0) node[anchor=north] {\scriptsize $u_{122}$};
\draw[] (0.4082 ,   0.8165  ,  0.4082) node[anchor=south] {\scriptsize $u_{1223}$};

              \coordinate (X) at (3,0,0) ;
              \coordinate (Y) at (0,2,0) ;
              \coordinate (Z) at (0,0,2) ;

                       \draw[-latex] (1.62,0,0) -- (X) node[anchor=west] {$X$};
              \draw[-latex] (0,1.13,0) -- (Y) node[anchor=west] {$Y$};
              \draw[-latex] (0,0,1.1) -- (Z) node[anchor=west] {$Z$};
            \end{scope}
          \end{tikzpicture}
        \end{center}
\caption{\footnotesize $\sigma$ tessellations} \label{sigmafig}
\end{minipage}
\end{figure}

\begin{lemma}[The Spherical Laws of Sines and Cosines]\label{sphcosine}
Given a spherical triangle with sides $a, b, c$ and angles $A,B,C$ opposite to side $a,b,c$ resp., then: 
\begin{itemize}
\item[i)] $\cos(c)  =  \cos(a)\cos(b) + \sin(a)\sin(b)\cos(C);$ 
\item[ii)] $\cos(C)  =  -\cos(A)\cos(B) + \sin(A)\sin(B)\cos(c)$ and 
\item[iii)] $\frac{\sin(a)}{\sin(A)} = \frac{\sin(b)}{\sin(B)}$.
\end{itemize}
\end{lemma}

\noindent
The sum of angles within a spherical triangle is between $\pi$ (as the area approaches zero) and $3\pi$ (as the triangle fills the whole sphere). The \emph{spherical excess} of a triangle is the sum of its angles minus $\pi$ radians.

\begin{theorem}[Girard's theorem]\label{girard}
The area of a spherical triangle is equal to its spherical excess.
\end{theorem}

\subsection{Online point placing on the unit sphere}

Our aim is to insert a sequence of `uniformly distributed' points onto $\Sph$ in an online manner. After placing a point, it cannot be moved in the future. Let $p_i$ be the $i$'th point thus inserted and let $S_i = \{p_1, p_2, \ldots, p_i\}$ be the configuration after inserting the $i$'th point. Teramoto et al.  introduced the gap ratio~\cite{Teramoto}, which defines a measure of uniformity for point samples and we use this metric (similarly to \cite{walcom17}). 

Let $\rho_{\text{min}}: 2^{\Sph} \to \R$ denote the minimal distance between a set of points $S_i$, defined by $\rho_{\text{min}}(S_i) = \min_{p, q \in S_i, p \neq q} \dist{p, q}$. Recall that notation $2^{\Sph}$ means the set of all points lying on the $2$-sphere $\Sph$.  Let $\rho^{\Sph'}_{\text{max}}: 2^{\Sph} \to \R$ denote the maximal spherical diameter of the largest empty circle centered at some point of $\Sph' \subseteq \Sph$ not intersecting the given set of points $S_i$, defined by 
$\rho^{\Sph'}_{\text{max}}(S_i) = \max_{p \in \Sph'}\min_{q \in S_i} 2\cdot \dist{p, q}$. We then define $\rho^{\Sph'}(S_i) = \frac{\rho^{\Sph'}_{\text{max}}(S_i)}{\rho_{\text{min}}(S_i)}$ to be the \emph{gap ratio} of $S_i$ over $\Sph'$. When $\Sph' = \Sph$ (i.e. when points can be placed anywhere on the sphere), we define that $\rho(S_i) = \rho^{\Sph}(S_i)$.

We denote an \emph{equilateral} spherical triangle as one for which each side has the same length. By Lemma~\ref{sphcosine} (the spherical law of cosines), having three equal length edges implies that an equilateral spherical triangle has the same three angles. By Theorem~\ref{girard} (Girard's theorem), each such angle is greater than $\frac{\pi}{3}$ (for an equilateral triangle of positive area). Let $\Delta \subseteq \Sph$ denote the set of all spherical triangles on the unit sphere.

Consider a spherical triangle $T \in \Delta$. We define a \emph{triangular dissection} function $\sigma: \Delta \to 2^{\Delta}$ in the following way. If $T \in \Delta$ is defined by $T = \langle u_1, u_2, u_3\rangle$, then $\sigma(T) = \{T_1, T_2, T_3, T_4\} \subset \Delta$, where $T_1 = \langle u_1, u_{12}, u_{13}\rangle$, $T_2 = \langle u_{12}, u_{2}, u_{23}\rangle$, $T_3 = \langle u_3, u_{13}, u_{23}\rangle$ and $T_4 = \langle u_{12}, u_{13}, u_{23}\rangle$, with $u_{ij}$ being the midpoints (on the unit sphere) of the arc connecting $u_i$ and $u_j$ (see Fig.~\ref{sigmafig}). Define $\sigma_E(T)$ as the set of nine induced edges: $\{(u_1, u_{12}), (u_{12}, u_{2}), (u_{2}, u_{23}),$ $(u_{23}, u_{3}),$ $(u_{3}, u_{13}),$ $(u_{13}, u_{1}),$ $(u_{13}, u_{23}),$ $(u_{23}, u_{12}),$ $(u_{12}, u_{13})\}$.

We extend the domain of $\sigma$ to sets of spherical triangles so that $\sigma(\{T_1, T_2, \ldots, T_k\}) = \{\sigma(T_1), \sigma(T_2), \ldots, \sigma(T_k)\}$; thus $\sigma: 2^{\Delta} \to 2^{\Delta}$. Given a spherical triangle $T \in\Delta$, we then define that $\sigma^1(T) = \sigma(T)$ and $\sigma^k(T) = \sigma(\sigma^{k-1}(T))$ for $k > 1$. For notational convenience, we also define that $\sigma^0(T) = T$ (the identity tessellation). We similarly extend $\sigma_E(T)$ to a set of triangles: $\sigma_E(\{T_1, T_2, \ldots, T_k\}) = \{\sigma_E(T_1), \sigma_E(T_2), \ldots, \sigma_E(T_k)\}$ and let $\sigma_E^k(T) =  \sigma_E(\sigma^k(T_1), \sigma^k(T_2), \ldots, \sigma^k(T_k))$. See Fig.~\ref{sigmafig} for an example showing the tessellation of $T$ to depth $2$ (e.g. $\sigma^2(T)$) and the set of edges $\sigma_E^2(T)$.

Let $\mu: \Delta \to 2^{\Sph}$ be a function which, for an input spherical triangle, returns the (unique) set of three points defining that triangle. For example, given a spherical triangle $T = \langle p_1, p_2, p_3 \rangle$, then $\mu(T) = \{p_1, p_2, p_3\}$. Clearly $\mu$ may be extended to sets of triangles by defining that $\mu(\{T_1, \ldots, T_k\}) = \{\mu(T_1), \ldots, \mu(T_k)\}$; thus $\mu: 2^{\Delta} \to 2^{\Sph}$. When there is no danger of confusion, by abuse of notation, we sometimes write $T$ rather than $\mu(T)$. This allows us to write $\rho(T)$ (or $\rho(\sigma^k(T))$) for example, as the gap ratio of the three points defining spherical triangle $T$ (resp. the set of points in the $k$-fold triangular dissection $\sigma^k(T)$).

We will also require an ordering on the set of points generated by a tessellation $\sigma^k(T)$. Essentially, we wish to order the points as those of $\sigma^0(T) = T$ first (in any order), then those of $\sigma^1(T)$ in any order \emph{but omitting the points of $\sigma^0(T) = T$}, then the points of $\sigma^2(T)$, omitting points in triangles of $\sigma^0(T)$ or $\sigma^1(T)$ etc. To capture this notion, we introduce a function $\tau: \Delta \times \Z^+ \to 2^{\Sph}$ defined thus:
$$\tau(T,k) = \left\{ \begin{array}{ll} \mu(\sigma^k(T)) - \mu(\sigma^{k-1}(T))& ; \text{ if } k \geq 1 \\ \mu(T) & ;  \text{ if } k = 0 \end{array} \right.
$$ 

As an example, in Fig.~\ref{sigmafig},  
$\tau(T, 0)  =  \{u_1, u_2, u_3\}$,  $\tau(T, 1)  =  \{u_{12}, u_{13}, u_{23}\}$, and
$\tau(T, 2)  =  \{u_{112}, u_{122}, u_{232}, u_{323}, u_{133}, u_{113}, u_{1323}, u_{1213}, u_{1223}\}$.
By abuse of notation, we redefine $\sigma^k(T)$ such that $\sigma^k(T) = \tau(T, 0) \cup \tau(T, 1) \cup \cdots \cup \tau(T, k)$ is an ordered set, ordered by points $\tau(T, 0)$ first and $\tau(T, k)$ last (points within any $\tau(T, j)$ for $1 \leq j \leq k$ are in any order). 

\section{Overview of Online Vertex Insertion Algorithm}\label{overviewsec}

Our algorithm is a two stage strategy. In stage one, we project the $12$ vertices of the regular icosahedron onto the unit sphere. The first two points inserted should be opposite each other (antipodal points), but the remaining $10$ points can be inserted in any order, giving a stage one gap ratio of $\frac{\pi}{\arccos{\left( \frac{1}{\sqrt{5}}\right)}} \approx 2.8376$.

In the second stage, we treat each of the $20$ equilateral spherical triangles of the regular icosahedron in isolation. We show in Lemma~\ref{mainlem} that the gap ratio for our tessellation is `local' and depends only on the local configuration of vertices around a given point. This allows us to consider each triangle separately. During stage two, we use the fact that these twenty spherical triangles are equilateral and apply Lemma~\ref{gaprat} to independently tesselate each triangle recursively in order to derive an upper bound of the gap ratio in stage two of $
\frac{2(3-\sqrt{5})}{\arcsin{\left(\frac{1}{2}\sqrt{2-\frac{2}{\sqrt{5}}}\right)}} \approx 2.760
$.

We note here that the radius of the sphere does not affect the gap ratio of the point insertion problem, and thus we assume a unit sphere throughout.

The algorithmic procedure to generate an infinite set of points is shown in Algorithm~\ref{thealg}. To generate a set of $k$ points $\{p_1, p_{2}, \ldots, p_{k}\}$, we choose the first $k$ points generated by the algorithm. 

\hrulefill

\begin{algorithm}[H]
\caption{Placing infinitely many points on the unit sphere using our recursive tessellation procedure on the regular icosahedron.} \label{thealg}

\begin{algorithmic}
    \STATE \textbf{Stage one:} Project $12$ vertices of the icosahedron to the unit sphere:
    \bindent
	\STATE{Place two antipodal points on the unit sphere.}
	\STATE{Place the remaining ten points in any order.}
	\STATE{Arbitrarily label the $20$ minimal spherical triangles $T = \{T_1, \ldots, T_{20}\}$.}
    \eindent
    \STATE \textbf{Stage two:} Recursively tessellate minimal triangles
    \bindent
    	\STATE {Let $T' \leftarrow T$}
	\WHILE {TRUE}
		\FORALL{minimal spherical triangles $R \in T$}
   		\STATE {Let $T' \leftarrow (T' \cup \sigma(R)) - R$}
    		\ENDFOR
		\STATE {Let $T \leftarrow T'$}
	\ENDWHILE
    \eindent
\vspace{0.2cm}
\end{algorithmic}
\end{algorithm}

\hrulefill

\section{Gap ratio of equilateral spherical triangles}\label{gap}

We will require several lemmata regarding tessellations of spherical triangles. The following lemma is trivial from the spherical sine rule and Girard's theorem.

\begin{lemma}\label{eqieqilem}
Let $T \in \Delta$ be an equilateral triangle. Then the central triangle in the tessellation $\sigma(T)$ is also equilateral.
\end{lemma}
\begin{proof}
Consider Fig.~\ref{sigmafig}. The lemma claims that if $\langle u_1, u_2, u_3 \rangle$ is equilateral, then so is $\langle u_{12}, u_{13}, u_{23} \rangle$ (and therefore also $\langle u_{1213}, u_{1223}, u_{1323} \rangle$). 

As a consequence of the spherical cosine rule, an equilateral spherical triangle will have three equal interior angles, each of which is larger than $\frac{\pi}{3}$ (otherwise, by Girard's theorem, it has zero area). Since the edge lengths of $T$ are identical, then the central triangle of $\sigma(T)$ also has equal length edges, again by the spherical cosine rule.
\end{proof}

It is worth emphasising in Lemma~\ref{eqieqilem} that the other three triangles in the triangular dissection of an equilateral triangle are \emph{not} equilateral, and have a strictly smaller area than the central triangle. This \emph{deformation} of the recursive triangular dissection makes the analysis of the algorithm nontrivial. The following lemma equates the distance from the centroid of an equilateral spherical triangle to a vertex of that triangle.

\begin{lemma}\label{centroidlem}
Let $T = \langle u_1, u_2, u_3 \rangle \in \Delta$ be an equilateral triangle with centroid $u_c$ and edge length $\zeta(u_1, u_2) = \alpha$. Then $\zeta(u_1, u_c) = \zeta(u_2, u_c) = \zeta(u_3, u_c) = \arcsin{\left(2\sin(\frac{\alpha}{2})/\sqrt{3}\right)}$.
\end{lemma}

\begin{proof}
Consider Fig.~\ref{gaprat}. The centroid of $T$ is, as for standard triangles, the unique point $u_c$ of $T$ from which the (spherical) distance satisfies $\zeta(u_c, u_1) = \zeta(u_c, u_2) = \zeta(u_c, u_3)$. Let then $x = \zeta(u_c, u_1)$.  By the sine rule of spherical trigonometry:
$$
\frac{\sin(x)}{\sin(\pi/2)} = \frac{\sin(\alpha/2)}{\sin(2\pi/6)},
$$
and since $\sin{\left(\frac{\pi}{3}\right)} = \frac{\sqrt{3}}{2}$, then $x = \arcsin{\left(\frac{2\sin(\frac{\alpha}{2})}{\sqrt{3}}\right)}$.
\end{proof}

Given an equilateral spherical triangle $T$, we will also need to determine the maximal and minimal edge lengths in $\sigma_E^k(T)$ for $k \geq 1$, which we now show. 

\begin{lemma}\label{minmaxlem}
Let $T = \langle u_1, u_2, u_3 \rangle \in \Delta$ be an equilateral triangle such that $\alpha = \zeta(u_1, u_2) \in (0, \frac{\pi}{2}]$ and $k \geq 1$. Then the minimal length edge in $\sigma_E^k(T)$ is given by any edge lying on the boundary of $T$. The maximal length edge of $\sigma_E^k(T)$ is any of the edges of the central equilateral triangle of $\sigma^k(T)$.
\end{lemma}
\begin{proof}
Consider Fig.~\ref{figuretansec}. The lemma states that in $\sigma^2(T)$ shown, the \emph{shortest} length edge of $\sigma^2_E(T)$ is $(u_1, u_{113})$, or indeed any such edge on the boundary of triangle $\langle u_1, u_2, u_3 \rangle$. The lemma similarly states that the \emph{longest} edge of $\sigma^2_E(T)$  is edge $(u_{1213}, u_{1223})$, or indeed any edge of the central equilateral triangle $\langle u_{1213}, u_{1223}, u_{1323} \rangle$.

Consider now Fig.~\ref{eqlenFig} illustrating $T = \langle u_1, u_2, u_3 \rangle$. Point $u_{12}$ (resp. $u_{13}$) is at the midpoint of spherical edge $(u_1, u_2)$ (resp. $(u_1, u_3)$).  Let $\alpha = \zeta(u_1, u_2) = \zeta(u_2, u_3) = \zeta(u_1, u_3)$ be the edge length. The intersection of spherical edges $(u_2, u_{13})$ and $(u_1, u_3)$ forms a spherical right angle at $u_{13}$. We denote $y = \zeta(u_{13}, u_{12})$, thus $y$ is the edge length of the central equilateral triangle of $\sigma(T)$ (and $\frac{\alpha}{2} = \zeta(u_1, u_{13})$ is the edge length of the minimal length edge of one of the non-central triangles in $\sigma(T)$; note that this is the same for each such triangle).

\begin{figure}[t]
\begin{minipage}[c]{0.45\linewidth}
\begin{center}
          \begin{tikzpicture}[scale=1.3]
            \pgfmathsetmacro\R{1} 
            \fill[ball color=white!10, opacity=0.1] (0,0,0) circle (\R); 
            \tdplotsetmaincoords{65}{120}
            \begin{scope}[tdplot_main_coords, shift={(0,0)}]
              \coordinate (O) at (0,0,0);


\tdplotdefinepoints(0,0,0)(1, 0, 0)(0,1,0)
\tdplotdrawpolytopearc[thick]{\R}{anchor=north}{\scriptsize $u_{12}$}

\tdplotdefinepoints(0,0,0)(0,1,0)(0, 0, 1)
\tdplotdrawpolytopearc[thick]{\R}{anchor=west}{$\alpha$}

\tdplotdefinepoints(0,0,0)(1, 0, 0)(0,0,1)
\tdplotdrawpolytopearc[thick]{\R}{anchor=east}{\scriptsize $u_{13}$}

\tdplotdefinepoints(0,0,0)(0.7071, 0.7071, 0)(0.7071,0,  0.7071)
\tdplotdrawpolytopearc[blue]{\R}{anchor=east}{}

\tdplotdefinepoints(0,0,0)(1, 0, 0)(0.8165  ,  0.4082,    0.4082)
\tdplotdrawpolytopearc[blue]{\R}{anchor=east}{}

\tdplotdefinepoints(0,0,0)(0, 1, 0)(0.7071,0,  0.7071)
\tdplotdrawpolytopearc[red]{\R}{anchor=east}{}

\tdplotdefinepoints(1,0,0)(0.7571, -0.1, 0.7071)(0.9633 ,   0.0816  ,  0.0816)
\tdplotdrawpolytopearc[]{0.2}{anchor=west}{}

\tdplotdefinepoints(0,0,1)(0.2425    ,     0  ,  0.9701)(0  ,  0.2425  ,  0.9701)
\tdplotdrawpolytopearc[]{0.2}{anchor=west}{}

\tdplotdefinepoints(0,1,0)(0.1715  ,  0.9701  ,  0.1715)(0  ,  0.9701  ,  0.2425)
\tdplotdrawpolytopearc[]{0.2}{anchor=west}{}

\tdplotdefinepoints(0, 0, 0)(0.7028  ,  0.1104  ,  0.7028)(0.6028  ,  0.1104 ,   0.8028)
\tdplotdrawpolytopearc[]{1}{anchor=north}{}
\tdplotdefinepoints(0, 0, 0)(0.6028  ,  0.1104 ,   0.8028)(0.5942      ,   0 ,   0.8043)
\tdplotdrawpolytopearc[]{1}{anchor=north}{}

\tdplotdefinepoints(0, 0, 0)(0.8532 ,   0.3266  ,  0.3266)(0.8594  ,  0.2449,    0.3980)
\tdplotdrawpolytopearc[]{1}{anchor=north}{}
\tdplotdefinepoints(0, 0, 0)(0.7946   , 0.3266  ,  0.4680)(0.8594 ,   0.2449  ,  0.3980)
\tdplotdrawpolytopearc[]{1}{anchor=north}{}

\draw[] (1,-0.1,0) node[anchor=north] {$u_1$};
\draw[] (-0.1,1.1,0) node[anchor=north] {$u_2$};
\draw[] (0,0,1.2) node[anchor=west] {$u_3$};
\draw[] (0,0,0.9) node[anchor=north] {$\gamma$};
\draw[] (0.8165  ,  0.4082   , 0.4082) node[anchor=west] {$y$};

\draw[] (0.3827     ,    0 ,   0.9239) node[anchor=east] {\scriptsize $\frac{\alpha}{2}$};
\draw[] (0.9239      ,   0   , 0.3827) node[anchor=east] {\scriptsize $\frac{\alpha}{2}$};
\draw[] (0,0.8,0.25) node[] {\scriptsize $\frac{\gamma}{2}$};
\draw[] (1,0.15,0.33) node[] {\scriptsize$\frac{\gamma}{2}$};
\draw[] (0.8887 ,  0.2113  ,  0.5074) node[] {\scriptsize$\frac{y}{2}$};

              \coordinate (X) at (3,0,0) ;
              \coordinate (Y) at (0,2,0) ;
              \coordinate (Z) at (0,0,2) ;

                       \draw[-latex] (1.62,0,0) -- (X) node[anchor=west] {$X$};
              \draw[-latex] (0,1.13,0) -- (Y) node[anchor=west] {$Y$};
              \draw[-latex] (0,0,1.1) -- (Z) node[anchor=west] {$Z$};
            \end{scope}
          \end{tikzpicture}
        \end{center}
\caption{\footnotesize Max and min lengths of $\sigma$ tessellations of an equilateral triangle.} \label{eqlenFig}
\end{minipage}
\hspace{1.5cm}
\begin{minipage}[c]{0.45\linewidth}
\begin{center}
          \begin{tikzpicture}[scale=1.4]
            \pgfmathsetmacro\R{1} 
            \fill[ball color=white!10, opacity=0.1] (0,0,0) circle (\R); 
            \tdplotsetmaincoords{65}{120}
            \begin{scope}[tdplot_main_coords, shift={(0,0)}]
              \coordinate (O) at (0,0,0);


\tdplotdefinepoints(0,0,0)(1, 0, 0)(0.7071 ,   0.7071 ,        0)
\tdplotdrawpolytopearc[thick, blue]{\R}{anchor=north}{}
\tdplotdefinepoints(0,0,0)(0.7071  ,  0.7071     ,    0)(0,1,0)
\tdplotdrawpolytopearc[thick]{\R}{anchor=north}{}

\tdplotdefinepoints(0,0,0)(0,1,0)(0, 0, 1)
\tdplotdrawpolytopearc[thick]{\R}{anchor=west}{$u_{23}$}

\tdplotdefinepoints(0,0,0)(1, 0, 0)(0,0,1)
\tdplotdrawpolytopearc[thick]{\R}{anchor=east}{$u_{13}$}

\tdplotdefinepoints(0,0,0)(0,0, 1)(0.5774, 0.5774, 0.5774)
\tdplotdrawpolytopearc[dashed]{\R}{anchor=east}{}

\tdplotdefinepoints(0,0,0)(0,1, 0)(0.5774, 0.5774, 0.5774)
\tdplotdrawpolytopearc[dashed]{\R}{anchor=east}{}

\tdplotdefinepoints(0,0,0)(0.5774, 0.5774, 0.5774)(0.7071, 0.7071, 0)
\tdplotdrawpolytopearc[blue]{\R}{anchor=west}{$\beta$}

\tdplotdefinepoints(0,0,0)(1,0, 0)(0.5774, 0.5774, 0.5774)
\tdplotdrawpolytopearc[blue]{\R}{anchor=south}{$x$}

\tdplotdefinepoints(0,0,0)(0.5774, 0.5774, 0.5774)(0, 0.7071, 0.7071)
\tdplotdrawpolytopearc[blue, dashed]{\R}{anchor=south}{$\beta$}

\tdplotdefinepoints(1,0,0)(0.5774, 0.5774, 0.5774)(0.7071, 0.7071, -0.1)
\tdplotdrawpolytopearc[]{0.2}{anchor=west}{}
\tdplotdefinepoints(0.5774, 0.5774, 0.5774)(1, 0, 0)(0.7071, 0.7071, 0)
\tdplotdrawpolytopearc[]{0.2}{anchor=north}{}

\tdplotdefinepoints(0, 0, 0)(0.6993  ,  0.6993  ,  0.1483)(0.7914   , 0.5934   , 0.1469)
\tdplotdrawpolytopearc[]{1}{anchor=north}{}
\tdplotdefinepoints(0, 0, 0)(0.7914   , 0.5934   , 0.1469)(0.8   , 0.6   , 0)
\tdplotdrawpolytopearc[]{1}{anchor=north}{}

\draw[] (1,-0.1,0) node[anchor=north] {$u_1$};
\draw[] (0.1,1,0) node[anchor=north] {$u_2$};
\draw[] (0,0,1) node[anchor=east] {$u_3$};
\draw[] (0.7071 , 0.7071 , 0) node[anchor=north] {$u_{12}$};
\draw[] (0.5774,0.5774,0.5774) node[anchor=west] {$u_c$}; 
\draw[] (1.4,0.6,0.05) node[] {$\frac{\alpha}{2}$};
\draw[] (0.5774,0.13,-0.045) node[] {$\frac{\gamma}{2}$};
\draw[] (0.93, 0.66, 0.4) node[] {$\frac{2\pi}{6}$};
              \coordinate (X) at (3,0,0) ;
              \coordinate (Y) at (0,2,0) ;
              \coordinate (Z) at (0,0,2) ;

              \draw[-latex] (1.62,0,0) -- (X) node[anchor=west] {$X$};
              \draw[-latex] (0,1.13,0) -- (Y) node[anchor=west] {$Y$};
              \draw[-latex] (0,0,1.1) -- (Z) node[anchor=west] {$Z$};

            \end{scope}
          \end{tikzpicture}
        \end{center}
\caption{\footnotesize Centroid calculations} \label{gaprat}
\end{minipage}
\end{figure}

By the spherical sine rule, $\sin{\left(\frac{\gamma}{2}\right)} = \frac{\sin{\frac{\alpha}{2}}}{\sin{\alpha}}$, which is illustrated by triangle $\langle u_2, u_{13}, u_3 \rangle$. Thus $\sin(\frac{\gamma}{2}) = \frac{1}{2} \sec{\left(\frac{\alpha}{2}\right)}$. Here we used the (standard trigonometric) identity that $\frac{\sin{\left(\frac{x}{2}\right)}}{\sin{x}} = \frac{1}{2}\sec{\left(\frac{x}{2}\right)}$. Further, one can see by the spherical sine rule that  
$\sin{\left(\frac{y}{2}\right)} = \sin{\left(\frac{\alpha}{2}\right)}\cdot\sin{\left(\frac{\gamma}{2}\right)} = \frac{1}{2} \frac{\sin{\left(\frac{\alpha}{2}\right)}}{\cos{\left(\frac{\alpha}{2}\right)}} = \frac{1}{2} \tan{\left(\frac{\alpha}{2}\right)}$.
This implies $\zeta(u_{13}, u_{12}) = y = 2\arcsin{\left(\frac{1}{2}\tan{\frac{\alpha}{2}}\right)}$ which is larger than $\frac{\alpha}{2}$ for $\alpha \in (0, \frac{\pi}{2}]$. To prove this, let $f(\alpha) =  2\arcsin(\frac{1}{2}\tan{\frac{\alpha}{2}})$ then $\frac{\textup{d}f}{\textup{d}\alpha} = \frac{2}{\cos^2{(\frac{\alpha}{2})}\sqrt{4-\tan^2(\frac{\alpha}{2})}}$ as is not difficult to prove. 
Noting that if $\alpha \in (0, \frac{\pi}{2}]$, then $\cos^2{\left(\frac{\alpha}{2}\right)} \in [\frac{1}{2}, 1]$ and $\sqrt{4-\tan^2(\frac{\alpha}{2})} \in [2, \sqrt{3}]$, then  $\frac{\textup{d}f}{\textup{d}\alpha} > \frac{1}{2} = \frac{\textup{d}\frac{\alpha}{2}}{\textup{d}\alpha}$ and thus since $f(\alpha) = 0 = \frac{\alpha}{2}$ when $\alpha = 0$, then $y > \frac{\alpha}{2}$ for $\alpha \in (0, \frac{\pi}{2}]$.

For any depth $k$-tessellation $\sigma^k(T)$, the maximal edge length of $\sigma^k_E(T)$ will thus be given by the length of the edges of the central equilateral triangle and the minimal length edges will be located on the boundary of $T$ as required.
\end{proof}

Given an equilateral spherical triangle $T = \langle u_1, u_2, u_3\rangle$, we now consider the gap ratio implied by the restriction of points to those of $T$. The first part of this lemma shows that the gap ratio of a depth-$k$ tessellation is lower than the gap ratio of a depth-$k+1$ tessellation (when restricted to points of $T$), and the second part shows that in the limit, the upper bound converges.

\begin{lemma}\label{mainlem}
Let $T = \langle u_1, u_2, u_3\rangle$ be an equilateral spherical triangle with spherical edge length $\alpha$, then: 
\begin{enumerate}
\item[i)] $\rho^{T}(\mu(\sigma^k(T))) < \rho^{T}(\mu(\sigma^{k+1}(T)))$;
\item[ii)] $\lim_{k \to \infty} \rho^{T}(\mu(\sigma^k(T))) = \frac{4\sin{(\frac{\alpha}{2})}}{\alpha\sqrt{3-4\sin^2{(\frac{\alpha}{2})}}}$.
\end{enumerate}
\end{lemma}
\begin{proof}

Consider Fig.~\ref{gaprat} and let $\alpha = \zeta(u_1, u_2) = \zeta(u_2, u_3) = \zeta(u_1,u_3)$ be the edge length of the equilateral triangle $T = \langle u_1, u_2, u_3\rangle$. Let us calculate $\rho^T(\sigma^0(T)) = \rho^T(T)$. Note by abuse of notation that we write $\rho^T(T)$ rather than the more formal $\rho^T(\mu(T))$, as explained previously.  Recall then that $\rho^T(T)$ denotes the gap ratio of point set $\mu(T)$ when the maximal gap ratio calculation is restricted to points of $T$. 

We see that $\rho_{\text{min}}(T) = \alpha$ since all edge lengths of $T$ are identical. Clearly $\rho^T_{\text{max}}(T) = 2x$; in other words the maximal spherical diameter of the largest empty circle centered inside $T$ should be placed at the centroid $u_c$ of $T$. This follows since if the circle is centered at any other point of $T$, then it will be closer to at least one vertex of $T$ and therefore the maximal ratio would only decrease. Thus $\rho^T(\sigma^0(T)) = \frac{2x}{\alpha}$.

By Lemma~\ref{eqieqilem}, triangle $\langle u_{12}, u_{23}, u_{13} \rangle$ in the decomposition $\sigma(T)$ is also equilateral. It is clear that $\beta > \alpha/2$ in Fig.~\ref{gaprat} by the spherical sine rule, since $\gamma > \pi/3$ (by Girard's theorem). Therefore, $\rho_{\text{min}}(\sigma^1(T)) = \frac{\alpha}{2}$, $\rho^T_{\text{max}}(\sigma^1(T)) = 2\beta$ and thus $\rho^T(\sigma^1(T)) = \frac{4\beta}{\alpha}$. We now show that $\frac{2x}{\alpha} < \frac{4\beta}{\alpha}$, which is true if $x < 2\beta$. 

Let us consider the projection of equilateral  spherical triangle $\langle u_1, u_2, u_3 \rangle$ from the center of the unit sphere to a tangent plane at the point $u_c$. The point $u_c$ is the centroid of spherical triangle $\langle u_1, u_2, u_3 \rangle$, as well its  projection to the plane $P$, given by the planar triangle $\langle u'_1, u'_2, u'_3 \rangle$.

The median of spherical triangle $\langle u_1, u_2, u_3 \rangle$ has length $x+\beta$. The range of $x$ is from $\beta$ to $2\beta$. This follows from the fact that in the maximal equilateral spherical triangle case
(i.e. when each angle is $\pi$ and the triangle forms a half sphere) $\beta = x = \frac{\pi}{2}$ and when the area of the 
spherical triangle converges to zero, the median of the spherical triangle
$\langle u_1, u_2, u_3\rangle$ converges to the median of the triangle
projection $\langle u'_1, u'_2, u'_3\rangle$, and $x$ converges to $2 \beta$
as the centroid of a Euclidean triangle divides each median in the ratio $2:1$.

We thus see that the gap ratio of the (six) points of $\sigma^1(T)$ is greater than the gap ratio of the (three) points of $\sigma^0(T)$, when restricted to points of $T$. Since the maximal ratio is calculated by using a circle centered at the centroid of the triangle $T$, this argument applies recursively and for each tessellation $\sigma^k(T)$, the maximal ratio is given by twice the distance of the centroid to the vertices of the central equilateral triangle of the tesselation by Lemma~\ref{minmaxlem}, and therefore the gap ratio increases at each depth of the tessellation, which proves statement one of the lemma.  We will now determine $\lim_{k \to \infty} \rho^T(\sigma^k(T))$ to prove the second statement.

We may observe that $\rho_{\text{min}}(\sigma^k(T)) = \frac{\alpha}{2^k}$, since the outer edges of triangle $T$ (with length $\alpha$) are subdivided into two $k$ times under $\sigma^k(T)$ and all interior edges have greater length. As explained above, the maximal diameter circle which may be placed on a point of $T$ which does not intersect points of $\sigma^k(T)$ will be centered at the centroid $u_c$ of $T$ and have a diameter twice the distance from $u_c$ to a vertex of that triangle. 

Construct a plane $P_{u_c}$ tangent to the point $u_c$ (the centroid of $T$). In Fig.~\ref{gaprat}, note that $x = \zeta(u_1, u_c) = \arcsin{\left(\frac{2\sin(\frac{\alpha}{2})}{\sqrt{3}}\right)}$ by Lemma~\ref{centroidlem}. The distance\footnote{Note that here we refer to the Euclidean distance between the points, rather than the spherical distance, since the projected points are not on the sphere} from the centroid point to a vertex projected by points $u_1, u_2$ or $u_3$ is given by $\tan{x}$ (see Fig.~\ref{figuretansec} and Section~\ref{spher-trig}), thus the edges of the projection of triangle $T$ have length $y = \sqrt{3}\tan{x}$ since the projection of an equilateral triangle about its centroid from the origin to $P_{u_c}$ is equilateral, with the same centroid (and the distance from the centroid of a Euclidean triangle to any vertex is of course given by $\frac{e}{\sqrt{3}}$, where $e$ is the edge length of the triangle). The central tessellated triangle thus has edges whose length starts to approximate $\frac{\sqrt{3}\tan{x}}{2^k}$ in the limit, since as $k \to \infty$, then this triangle lies on the plane tangent at $u_c$ (i.e. the difference between the edge length of the spherical triangle and its projection decreases to zero). This implies that the  maximal spherical diameter of the largest empty circle centered at $u_c$ is the distance from $u_c$ to one of these vertices, which approaches $\frac{2\tan{x}}{2^k}$ in the limit as $k \to \infty$. Therefore, 
\begin{eqnarray}
\lim_{k \to \infty} \rho^{T}(\sigma^k(T)) & = & \frac{\rho^{T}_{\text{max}}(\sigma^k(T))}{\rho_{\text{min}}(\sigma^k(T))} = \frac{2\tan{x}/2^k}{ \alpha/2^k}
=  \frac{2\tan\left({\text{arcsin}\left(\frac{2\sin(\frac{\alpha}{2})}{\sqrt{3}}\right)}\right)}{\alpha} \label{eqq1} \\ 
&= & \frac{4\sin{(\frac{\alpha}{2})}}{\alpha\sqrt{3-4\sin^2{(\frac{\alpha}{2})}}} \label{eqq2}
\end{eqnarray}
Moving from (\ref{eqq1}) to (\ref{eqq2}), we used the identity $\tan{(\arcsin{x})} = \frac{x}{\sqrt{1-x^2}}$.
\end{proof}

\section{Regular icosahedral tessellation}\label{regular}

As explained in Section~\ref{overviewsec} and Algorithm~\ref{thealg}, our algorithm consists of two stages. Using the lemmata of the previous section, we are now ready to show that the stage one gap ratio is no more than $\frac{\pi}{\arccos{\left( \frac{1}{\sqrt{5}}\right)}} \approx 2.8376$  and the second stage gap ratio is no more than $
\frac{2(3-\sqrt{5})}{\arcsin{\left(\frac{1}{2}\sqrt{2-\frac{2}{\sqrt{5}}}\right)}} \approx 2.760
$.

\begin{lemma}\label{phaseone}
The gap ratio of stage one is no more than $\frac{\pi}{\arccos{\left( \frac{1}{\sqrt{5}}\right)}} \approx 2.8376$.
\end{lemma}
\begin{proof}
The points of a regular icosahedron can be defined by taking circular permutations of $(0, \pm 1, \pm \phi)$, where $\phi = \frac{1 + \sqrt{5}}{2}$ is the golden ratio. Let $V'$ be the set of the twelve such vertices. Normalising each element of $V'$ gives a set $V$. Note that the area of each spherical triangle is given by $\frac{\pi}{5}$ since we have a unit sphere and twenty identical spherical triangles forming a tessellation. By Girard's theorem (Theorem~\ref{girard}), this implies that $3\gamma - \pi = \frac{\pi}{5}$, where $\gamma$ is the interior angle of the equilateral triangle and thus $\gamma = \frac{2\pi}{5}$. By the second spherical law of cosines (Lemma~\ref{sphcosine}), this implies that the spherical distance between adjacent vertices, $\alpha$, is thus given by $\cos(\alpha)  =  \frac{\cos(\frac{2\pi}{5}) + \cos^2(\frac{2\pi}{5})}{\sin^2(\frac{2\pi}{5})}  =  \frac{\frac{1}{4}(\sqrt{5}-1) + \frac{3}{8}+ (\frac{\sqrt{5}}{8})}{\frac{5}{8} + \frac{\sqrt{5}}{8}} = \frac{\frac{1}{8}(1 + \sqrt{5})}{\frac{1}{8}(5 + \sqrt{5})} = \frac{1}{\sqrt{5}}$ 
and therefore $\alpha \approx 1.1071$. The first two points are placed opposite to other, for example $u_1 \approx (0, -0.5257, 0.8507)$ and $u_2 \approx (0, 0.5257, -0.8507)$. At this stage, the gap ratio is $1$, since the largest circle may be placed on the equator (with $u_1$ and $u_2$ at the poles) with a diameter of $\pi$, whereas the spherical distance between $u_1$ and $u_2$ is $\pi$. The remaining ten vertices of the normalised regular icosahedron are placed in any order. The minimal distance between them is given by $\alpha$ above, and thus the gap ratio during stage one is no more than 
$\frac{\pi}{\arccos{\left( \frac{1}{\sqrt{5}}\right)}} \approx 2.8376,$
as required.
\end{proof}

As explained in Section~\ref{overviewsec} and Algorithm~\ref{thealg}, we start with the twenty equilateral spherical triangles produced in stage one, denoted the depth-$0$ tessellation of $\Sph$. We apply $\sigma$ to each such triangle to generate $20*4 = 80$ smaller triangles (note that not all such triangles are equilateral, in fact only eight triangles at each depth tessellation are equilateral). At this stage we have the depth-$1$ tessellation of $\Sph$. We recursively apply $\sigma$ to each spherical triangle at depth-$k$ to generate the depth-$k+1$ tessellation, which contains $20*4^{k+1}$ spherical triangles.

\begin{lemma}\label{phasetwo}
The gap ratio of stage two is no more than $ \frac{12-4\sqrt{5}}{\arccos{\left(\frac{1}{\sqrt{5}}\right)}}\approx 2.760$.
\end{lemma}
\begin{proof}
At the start of stage two, we have $20$ equilateral spherical triangles, $T_1, \ldots, T_{20}$ which are identical (up to rotation). Note that moving from depth-$k$ tessellation to depth-$k+1$, each edge of $\sigma_E^k(T_i)$ will be split at its midpoint, therefore applying $\sigma$ to any spherical triangle will only `locally' change the gap ratio of at most two adjacent triangles. Thus the order in which $\sigma$ is applied to each triangle at depth $k$ is irrelevant.

Assume that we have a (complete) depth-$k$ tessellation with $20*4^{k}$ triangles.  Lemma~\ref{mainlem} tell us that the gap ratio increases from the depth-$k$ to the depth-$k+1$ tessellations for all $k \geq 0$ and in the limit, the gap ratio of the depth-$k$ tessellation of each $T_i$ is given by:
\begin{equation}
\lim_{k \to \infty} \rho^T( \sigma^k(T_i))  =  \frac{4\sin{(\frac{\alpha}{2})}}{\alpha\sqrt{3-4\sin^2{(\frac{\alpha}{2})}}} \label{stage2bound}
\end{equation}
where $\alpha$ is the length of the edges of $T_i$ (i.e. the length of those triangles produced by stage $1$ via the icosahedron). When we start to tessellate the depth-$k$ spherical triangles, until we have a complete depth-$k+1$ tessellation, applying $\sigma$ to each of the triangles may decrease the minimal gap ratio at most by a factor up to $2$ overall (since we split each edge at its midpoint). The maximal gap ratio cannot increase, but decreases upon completing the depth $k+1$ tessellation. Therefore, we multiply Eq.~(\ref{stage2bound}) by $2$ to obtain an upper bound of the gap ratio for the entire sequence, not only when some depth-$k$ tessellation is complete. We now solve Eq.~(\ref{stage2bound}), after multiplying by $2$, by substituting $\alpha =  \arccos{\left( \frac{1}{\sqrt{5}}\right)} \approx 1.1071$. This is laborious, but by noting that $\sin(\frac{\alpha}{2}) = \sqrt{\frac{1}{10}(5-\sqrt{5})}$ and $\sin^2(\frac{\alpha}{2}) =\frac{1}{10}(5-\sqrt{5})$, then:
$
\frac{8\sin{(\frac{\alpha}{2})}}{\alpha\sqrt{3-4\sin^2{(\frac{\alpha}{2})}}} =  \frac{8\sqrt{\frac{1}{10}(5-\sqrt{5})}}{  \arccos{\left(\frac{1}{\sqrt{5}}\right)}\sqrt{1+\frac{2}{\sqrt{5}}}  } =
 \frac{12-4\sqrt{5}}{\arccos{\left(\frac{1}{\sqrt{5}}\right)}}\approx 2.760
$. 
Therefore, the gap ratio during stage two is upper bounded by $2.76$.
\end{proof}

\begin{theorem}
The gap ratio of the icosahedral triangular dissection is equal to $\frac{\pi}{\arccos{\left( \frac{1}{\sqrt{5}}\right)}} \approx 2.8376$.
\end{theorem}
\begin{proof}
This is a corollary of Lemma~\ref{phaseone} and Lemma~\ref{phasetwo}. Stage one of the algorithm has a larger gap ratio than stage two in this case.
\end{proof}

We can now prove the first nontrivial lower bound when we have only 2 or 3 points on the sphere in this online version of the problem.

\begin{theorem}\label{lowerboundthm1}
The gap ratio for the online problem of placing the first three points on the sphere cannot be less than $\frac{1+\sqrt{5}}{2} \approx 1.6180$.
\end{theorem}

\begin{proof}
Let us first estimate the ratio with only two points when one point is located at the north pole of the $2$-sphere and the other is shifted by a distance $x$ 
from the south pole. The gap ratio in this case is $\frac{\pi+x}{\pi-x}$, which increases from $1$ to $\infty$ when $x \geq 0$.

Let us now consider the case with three points. 
If we place the third point on the plane $P$ defined by the center of the sphere and the other two points,
then the gap ratio will be $\frac{\pi}{\frac{\pi+x}{2}}$=$\frac{2\pi}{\pi+x}$ as in this case the maximal diameter of an empty circle is $\pi$ regardless of the position of the third point on $P$.
Note that the center of the largest circle not intersecting points $p_1, p_2, p_3$ will be orthogonal to these three points (giving the diameter of the largest empty circle as at least $\pi$) and the minimal distance between two points can be maximised by positioning the third point at the largest distance from the initial two points so that the closest two points have a distance of $\frac{\pi+x}{2}$.

If the third point is not on the plane $P$ then the ratio will be equal to some value 
$\frac{a}{b}$ that is  larger than $\frac{2 \pi}{\pi+x}$.
This follows from the fact that the value $a$ which is the maximal gap would be greater than $\pi$ and the minimal gap $b$ would be 
less than $\frac{\pi+x}{2}$. So the minimal gap ratio that can be achieved for the three points will be represented by the expression $\frac{2 \pi}{\pi+x}$. See Fig.~\ref{threePointEg}.

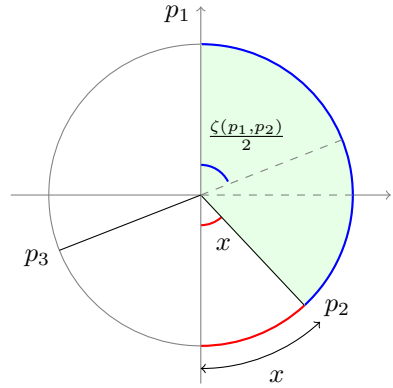
\begin{figure}
\begin{center}          
\begin{tikzpicture}[scale=1]

\draw[gray,->] (-2.5,0) -- (2.5,0);

\draw [thin, gray] (0,0) circle (2cm);
\draw [thick, blue, fill=ggcol] (0,0) ++(90:2) arc (90:-47:2);
\draw [ggcol, fill=ggcol] (0,0) -- (0, 2) -- (1.364, -1.463) -- cycle;
\draw [thick, red] (0,0) ++(-90:0.4) arc (-90:-47:0.4);
\draw [thick, red] (0,0) ++(-90:2) arc (-90:-47:2);
\draw [thick, blue] (0,0) ++(90:0.4) arc (90:26.5:0.4);
\node at (0.6,0.8) {{\small $\frac{\zeta(p_1, p_2)}{2}$}};
\draw (0,0) -- (1.364, -1.463);

\node[] at (0.3,-0.65) {$x$};
\node[] at (1,-2.4) {$x$};
\node[] at (-0.3, 2.4) {$p_1$};
\node[] at (1.8, -1.5) {$p_2$};
\node[] at (-2.15, -0.833) {$p_3$};
\node[] at (1.578, -2.2) {};
\draw[gray,->] (0,-2.5) -- (0,2.5);
\draw[gray,dashed] (0,0) -- (1.86,0.733);
\draw (0,0) -- (-1.86, -0.733);
\draw[gray,dashed] (0,0) -- (2,0);
\draw [<->] (0,0) ++(-90:2.3) arc (-90:-47:2.3);
\end{tikzpicture}
\caption{Optimal placement of first three points to minimize gap ratio.}\label{threePointEg}
\end{center}
\end{figure}

By solving the equation where the left hand side represents the gap ratio in the case of 3 points (a decreasing function) and the right hand side representing the case with 2 points (an increasing function), we find a positive value of the one unknown $x$:
$\frac{2 \pi}{\pi+x}=\frac{\pi+x}{\pi-x}$. The only positive value $x$ satisfying the above equation has the value  $\pi(\sqrt{5}-2)$ and the gap ratio for this value $x$
is equal to $\frac{1+\sqrt{5}}{2}$.
\end{proof}

We now show that by considering the first four points, we can do slightly better in lower bounding the minimal gap ratio.

\begin{theorem}\label{lowerboundthm2}
The gap ratio for the online problem of placing four or more points on the sphere is greater than $1.726$.
\end{theorem}

\begin{proof}
Consider Fig.~\ref{figpos234}. As in the previous proof, we begin by placing the first point $u_1$ at the north pole and the second point $u_2$ at some distance $x$ from the south pole on the $z=0$ plane, without loss of generality. We use an ansatz to estimate that for four points, the gap ratio cannot be less than $1.726$. In order that the gap ratio for the first two points is not greater than $1.726$, we have that:
\begin{eqnarray*}
\frac{\pi + x}{\pi - x} & \leq & 1.726 \\
\Rightarrow x & \leq & \frac{0.726}{2.726}
\end{eqnarray*}

To minimize the gap ratio after placing the third and fourth points, $u_3$ should be placed on the great circle $A$ lying equidistantly between $u_1$ and $u_2$. shown in Fig.~\ref{figpos234}. To see this, assume that $u'_3$ is not placed on $A$ and consider the plane $P$ which intersects points $u_1$, $u_2$ and $u_3$ as shown in Fig.~\ref{figureplanesph}. Let $u_3$ be the point lying on the intersection between $P$ and $A$ in the same hemisphere as $u'_3$. The center of the largest empty circle of $S$ after placing $u_1$, $u_2$ and $u'_3$ is the optimal position for point $u_4$ since placing $u_4$ at this point minimizes the largest remaining empty circle and maximizes the distance to the closest existing point (shown by $\rho_{\text{max}}^{\mathcal{S}}(u_1, u_2, u_3)$ in Fig.~\ref{figureplanesph}). The center of the largest empty circle after placing $u_1$, $u_2$ and $u'_3$ is either at the antipodal point of the circumcenter of $\langle u_1, u_2, u_3 \rangle$ if $u_3$ lies outside the circle on $S$ containing $u_1, u_2$, or else at the antipodal point of the midpoint of $u_1, u_2$ if $u_3$ lies within this circle (i.e. if $u_3$ is in the top green segment of Fig.~\ref{figureplanesph} where the plane $P$ is then defined as intersecting $u_1, u_2$ and parallel to the tangent plane of the midpoint of $u_1$ and $u_2$). In the second case, the position of $u_3$ has no impact on the maximal empty circle and thus we may place $u'_3$ on the arc $A$ without loss of generality. We therefore consider the first case when $u'_3$ is not within the circle containing $u_1, u_2$ and denote by $u_4$ the position of the center of the largest empty circle, at the antipodal point of the circumcenter of $u_1, u_2, u_3$ as is shown in Fig.~\ref{figureplanesph}. Note that this point is identical for either $\langle u_1, u_2, u'_3 \rangle$ or $\langle u_1, u_2, u_3 \rangle$ since $u'_3$ and $u_3$ lie on the same plane $P$ and thus share the same circumcircle. The spherical distance from $u_3$ to the closest of $u_1$ or $u_2$ is clearly smaller than the distance of $u'_3$ to either of $u_1$ or $u_2$ by the definition of $A$ (again shown in Fig.~\ref{figureplanesph}) and thus the gap ratio of sequence $u_1, u_2, u_3, u_4$ is no larger than the gap ratio of $u_1, u_2, u'_3, u_4$. We therefore assume that $u_3$ lies on arc $A$ and $u_4$ is placed at the antipodal point of the center of the circumcircle of $\langle u_1, u_2, u_3\rangle$.

\begin{figure}[ht]
\begin{minipage}[c]{0.45\linewidth}
\begin{center}
          \begin{tikzpicture}[scale=1]
\filldraw[fill=ggcol, draw=blue] (1.7321,1.0) arc (30:150:2.0) -- cycle;
\draw[draw=blue, <-] (-1.9773, 0.9644) arc (154:360:2.2);
\draw[draw=blue, ->] (2.2,0.0) arc (0:26:2.2);
\draw [black] (-3.0 , 1.0) -- (3.0, 1.0);

\draw[gray,->] (-2.5,0) -- (2.5,0);
\draw[gray,->] (0,-2.5) -- (0,2.5);
\draw[gray] (-1.7321, 1.0) -- (1.7321, 1.0);

\draw [thin, gray] (0,0) circle (2cm);

\node[] at (-1.5,0.7) {$u_1$};
\node[] at (1.5,0.7) {$u_2$};
\node[] at (0.4,1.2) {$u_3$};
\node[] at (-0.5,1.3) {$u'_3$};
\node[] at (0.3,-0.3) {$O$};
\node[] at (0.35,-1.6) {$u_4$};
\node[] at (2.4,-2.1) {$\rho_{\text{max}}^{\mathcal{S}}(u_1, u_2, u_3)$};
\node[] at (0.45,0.25) {$\pi-x$};
\node[] at (-0.2,-1.0) {$A$};
\node[] at (2.0,1.2) {$P$};
\node[] at (0.5,2.2) {$u_{123}$};

\draw[gray,<->] (-1.43,0.4) -- (1.43,0.4);
\draw[black, thick] (0.0,-2.0) -- (0.0,2.0);

\draw[black,fill=red] (1.43, 1.0) circle (.6ex);
\draw[black,fill=red] (-1.43, 1.0) circle (.6ex);
\draw[black,fill=red] (0.0, 1.0) circle (.6ex);
\draw[black,fill=red] (0.0, -2.0) circle (.6ex);
\draw[black,fill=black] (-0.5, 1.0) circle (.3ex);
\draw[black,fill=black] (0.0, 2.0) circle (.3ex);

          \end{tikzpicture}
        \end{center}
\caption{\footnotesize Plane $P$ intersecting $u_1, u_2, u_3$ and optimal position for $u_4$. $A$ is the great circle equator of $u_1, u_2$.}\label{figureplanesph}
\end{minipage}
\hspace{1.5cm}
\begin{minipage}[c]{0.45\linewidth}
         \begin{center}
          \begin{tikzpicture}[scale=2]

\filldraw[fill=rrcol, draw=blue] (0.7419, -0.6706) arc (-42:90:1) to[out=-120,in=1200] (0.0, -0.405) to[out=-60,in=-130] (0.7419, -0.6706);

\draw[gray,thick, ->] (0.9139, 0.4059) to[out=-100,in=-10] (-0.9139, -0.4059); 

\draw[gray, dashed, ->] (0.71, -0.03) to[out=-65,in=80] (0.7419, -0.6706); 
\draw[gray, dashed, ->] (0.71, -0.03) to[out=90,in=335] (0, 1); 

\draw[gray, dashed] (-0.71, -0.13) to[out=110,in=200] (0, 1); 
\draw[gray, dashed] (-0.71, -0.13) to[out=200,in=100] (-0.9139, -0.4059); 
\draw[gray, dashed] (-0.71, -0.13) to[out=270,in=225] (0.7419, -0.6706); 


\draw[thin, gray] (0,0) circle (1cm);

\draw[thin, gray, <->] (0,-1.1) arc (-90:-45:1.1); 

\node[] at (0, 1.13) {$u_1$};
\node[] at (0.9, -0.7) {$u_2$};
\node[] at (0.05, -0.305) {$u_3$};
\node[] at (-0.3, -0.3) {$A$};
\node[] at (0.45, -0.35) {$z$};
\node[] at (-0.57, -0.13) {$u_4$};
\node[] at (0.5, 0) {$u_{123}$};
\node[] at (0.56, -1.12) {$x$};
\node[] at (1.1, 0.4059) {$u_{12}$};

\draw[black,fill=red] (0, 1.0) circle (.3ex); 
\draw[black,fill=black] (0, -1.0) circle (.15ex); 
\draw[black,fill=red] (0.7419, -0.6706) circle (.3ex);
\draw[black,fill=red] (0.0, -0.405) circle (.3ex); 
\draw[black,fill=black] (0.71, -0.03) circle (.15ex); 
\draw[gray,fill=gray] (-0.71, -0.13) circle (.25ex); 
          \end{tikzpicture}
        \end{center}
\caption{\footnotesize Optimal positions of $u_3$ and $u_4$ for fixed $u_1$ and $u_2$. Point $u_{123}$ is the circumcenter of $u_1, u_2, u_3$ and $u_{12}$ is the midpoint of $u_1$ and $u_2$.}\label{figpos234}
\end{minipage}
\end{figure}

Let us thus assume that $x = \frac{0.726}{2.726}$ and fix $u_1$ and $u_2$. Consider again Fig.~\ref{figpos234}. As $z$ (the distance along the great arc $A$) increases from $0$ (when $u_3$ is at the midpoint of $u_1$ and $u_2$) to $\pi$ (at the antipodal point of the midpoint of $u_1$ and $u_2$), the gap ratio of $\{u_1, u_2, u_3\}$ decreases, since the minimal distance between points only increases and the size of the maximal empty circle (located either at the antipodal point of the center of the circumcircle of $\langle u_1, u_2, u_3\rangle$ or else at the antipodal point of the midpoint of $u_1, u_2$) monotonically decreases.

\begin{figure}[ht]
        \makebox[\textwidth][c]{\includegraphics[width=17cm]{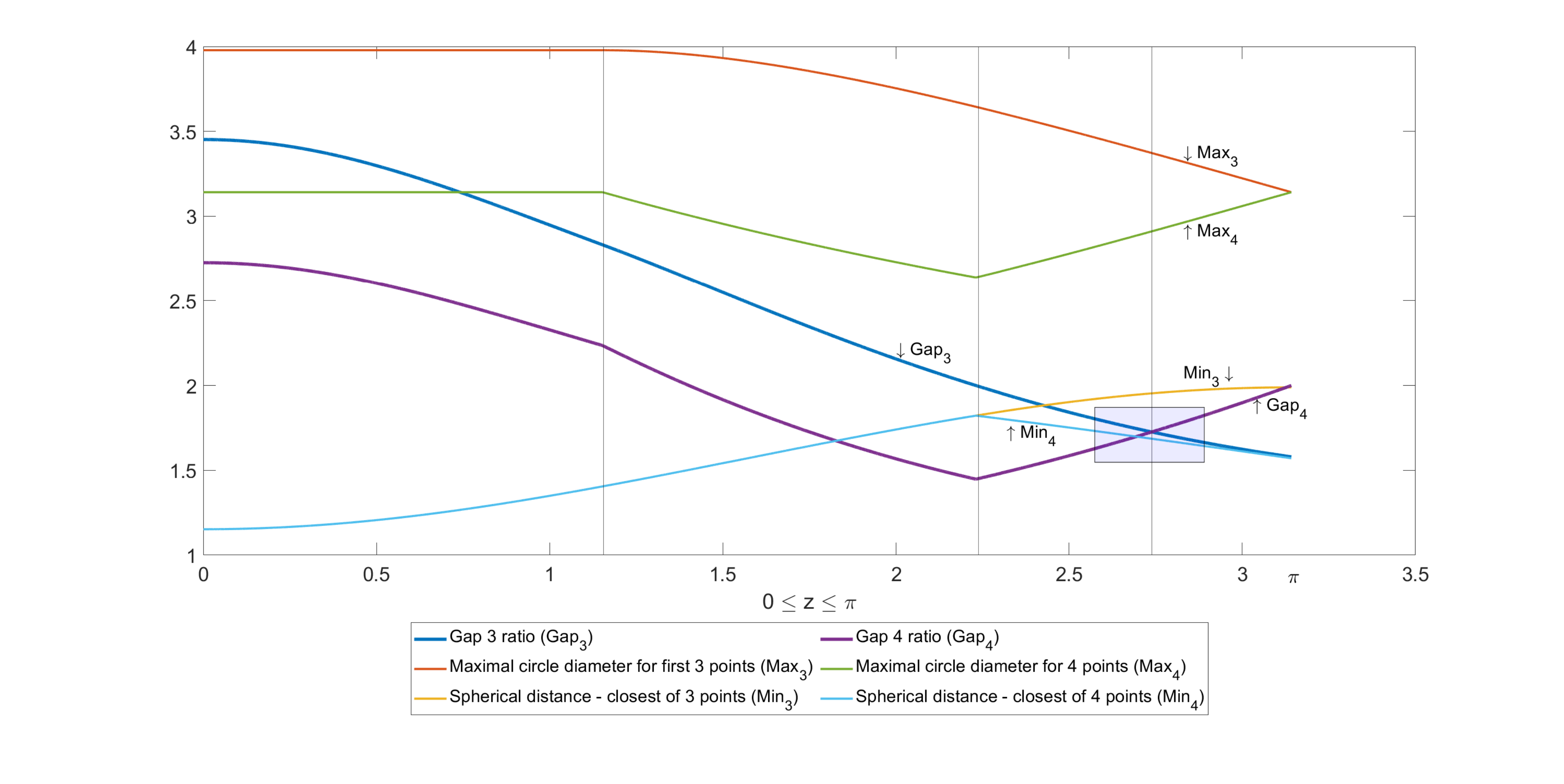}}%
   \caption{Gap-3 and Gap-4 ratios for various $z$ when $x = \frac{0.726}{2.726}\pi$.}\label{gap3gap4}%
\end{figure}

Let us denote by the \emph{$4$-gap ratio} the gap ratio of the four points $u_1, u_2, u_3, u_4$ ignoring the gap ratios of the subsequences $u_1, u_2$ and $u_1, u_2, u_3$ (similarly for the $3$-gap ratio, which ignores the gap ratio of $u_1, u_2$). The $4$-gap ratio of $u_1, u_2, u_3, u_4$ for $z$ in the range $[0, \pi]$ can be split into four segments as can be seen in Fig.~\ref{gap3gap4}. We will show that the optimal gap ratio occurs at the unique point where the gap-$3$ and gap-$4$ ratios cross.

Since we have shown that the gap 3 ratio is monotonically decreasing, decreasing $z$ from the intersection point of gap-$3$ and gap-$4$ ratios only only increases the overall gap ratio (which is the maximum of the two). At the intersection point, we have the following four points:
\[
u_1 = \begin{pmatrix} 0 \\ 1 \\ 0 \end{pmatrix}, 
u_2 = \begin{pmatrix}  0.742 \\ -0.670\\ 0 \end{pmatrix}, 
u_3 = \begin{pmatrix} -0.841\\ -0.374\\ 0.392 \end{pmatrix}, 
u_4 = \begin{pmatrix}  -0.259\\  -0.115\\ -0.959 \end{pmatrix}, 
\]

As $z$ increases from this intersection point, the gap-$4$ ratio increases which we now prove. Consider again Fig.~\ref{figureplanesph}. As $z$ increases, the plane $P$ moves downwards towards the origin since the circumcircle of $\langle u_1, u_2, u_3 \rangle$ has an increasing diameter. Once $\rho_{\text{max}}^{\mathcal{S}}(u_1, u_2, u_3) = 2\zeta(u_4, u_1) < \frac{2\pi}{3}$, then the largest empty circle of   $\{u_1, u_2, u_3, u_4\}$ is located at the antipodal point of $u_4$, i.e. the centre of the circumcircle of $\langle u_1, u_2, u_3\rangle$, which we denote $u_{123}$ in Fig.~\ref{figureplanesph}. Note that by the definition of circumcircle $\zeta(u_1, u_{123}) = \zeta(u_2, u_{123}) = \zeta(u_3, u_{123})$ and note also that $\zeta(u_1, u_{123}) + \zeta(u_1, u_{4}) = \pi$, as is clear from the Fig.~\ref{figureplanesph}.

The closest two points of $\{u_1, u_2, u_3\}$ is given by $\zeta(u_1, u_3) = \zeta(u_2, u_3)$. This follows since $\zeta(u_1, u_2) > \frac{2\pi}{3}$ and thus for any $z \in [0, \pi]$ then $\zeta(u_1, u_3) = \zeta(u_2, u_3) < \zeta(u_1, u_2)$. Since $u_4$ is at the antipodal point of the circumcentre of $\langle u_1, u_2, u_3 \rangle$, then $\zeta(u_1, u_4) = \zeta(u_2, u_4) = \zeta(u_3, u_4)$. Thus, once $\zeta(u_4, u_1) < \zeta(u_3, u_1)$ then we know that
\[ \zeta(u_4, u_1) = \zeta(u_4, u_2) = \zeta(u_4, u_2) < \zeta(u_1, u_3) = \zeta(u_2, u_3) < \zeta(u_1, u_2)
\]
For $u_1, u_2, u_3, u_4$ given above, the closest two points are therefore given by $\zeta(u_1, u_4) = 1.686 < \frac{2\pi}{3}$ since $\zeta(u_1, u_4) < \zeta(u_1, u_3) = 1.954 < \zeta(u_1, u_2) = 2.305$. The largest empty circle of $\{u_1, u_2, u_3, u_4\}$ is thus given by $w_{123}$. Finally, consulting Fig.~\ref{figureplanesph} once more, since increasing $z$ moves the plane $P$ towards the origin, we decrease $\zeta(u_1, u_4)$ and increase the maximum empty circle diameter given by $\zeta(u_1, u_{123})$ and thus the $4$-gap ratio necessarily increases. Therefore, the optimal gap ratio for fixed $u_1$ and $u_2$ and for $z \in [0, \pi]$ is at this intersection point, corresponding to points $u_1, u_2, u_3, u_4$ given above.

The gap ratios of $\{u_1, u_2, u_3, u_4\}$  can be calculated as: $\rho(u_1, u_2) = 1.726$, $\rho(u_1, u_2, u_3) = 1.7261$ and $\rho(u_1, u_2, u_3, u_4) = 1.7261$. Note that $u_3$ is placed on the arc $A$ equidistant from $u_1$ and $u_2$ (since $\zeta(u_1, u_3) = \zeta(u_2, u_3) =  1.9538$) as required, and that $u_4$ is in the opposite hemisphere from $u_3$ with $\zeta(u_1, u_4) = \zeta(u_2, u_4) = \zeta(u_3, u_4) = 1.686$. Now, since the gap ratios of points $\{u_1, u_2, u_3\}$ and $\{u_1, u_2, u_3, u_4\}$ are equal to $1.7261$, and as pointed out earlier, as $z$ increases the gap ratio of $\{u_1, u_2, u_3, u_4\}$ increases, and as $z$ decreases the gap ratio of $\{u_1, u_2, u_3\}$ increases, then for fixed $u_1$ and $u_2$, these points $u_3$ and $u_4$ are optimal.

It remains to reason about out choice of $x$ (the spherical distance from $u_1$ to $u_2$). We know that $x \leq \frac{0.726}{2.726}\pi$, otherwise the gap ratio of the first two points is already above the bound of $1.726$ and we have now analysed the case where we have equality. 

Consider Fig.~\ref{figpos234}. For any $x' \in [0, \frac{0.726}{2.726}\pi]$, there exists some choice of $z_2$ such that $\zeta(u_3, u_1) = \zeta(u_3, u_4)$. This is clear from the intermediate value theorem. The center of the largest empty circle for $u_1, u_2, u_3$ is at the antipodal point of the circumcenter of $u_1, u_2, u_3$, denoted $u_4$. Since $\zeta(u_3, u_1) = \zeta(u_3, u_4)$, then 
\[
\zeta(u_2, u_3)  = \zeta(u_1, u_3) = \zeta(u_3, u_4) = \zeta(u_1, u_4) = \zeta(u_2, u_4)
\]
The first equality is due to $u_3$ being placed on arc $A$ and the latter two equalities are due to $u_4$ being at the antipodal point of the circumcenter of $u_1, u_2, u_3$. Since $\zeta(u_3, u_4) =  \zeta(u_1, u_3)$, then the gap ratio is given by $\frac{2\zeta(u_3, u_4)}{\zeta(u_3, u_1)} = 2$. The numerator is derived since $u_4$ is the center of the largest empty circle for $u_1, u_2, u_3$ and the denominator denotes the closest two points. For larger values of $z$, the maximal empty circle diameter of the first three points decreases and the center of the largest empty circle is located at the circumcenter of the first three points ($u_{123}$ in Fig.~\ref{figureplanesph} and Fig.~\ref{figpos234}). For any choice of $x \leq \frac{0.726}{2.726}\pi$, if the gap ratio is to be less than $2$ we may thus assume that $z$ is chosen so that any increase in $z$ increases $\zeta(u_1, u_3)$ and the center of the largest empty circle is located at $u_{123}$, the circumcenter of $\langle u_1, u_2, u_3\rangle$.

Consider Fig.~\ref{figureplanesph} and assume the figure corresponds to the optimal choice of $z$ for $x = \frac{0.726}{2.726}\pi$. As $x$ decreases, the distance between points $u_1$ and $u_2$ increases and the distance from point $u_3$ to $u_1$ and $u_2$, for the same $z$ value, will decrease. In order that the gap-$3$ ratio of the first three points does not increase, the plane $P$ needs to move towards the origin, since the distance $\zeta(u_1, u_3)$ has decreased as $x$ decreased. Now consider the $4$-gap ratio. Since $P$ has moved towards the origin, the minimal distance between points can only decrease. The center of the maximal empty circle after placing $u_1, u_2, u_3, u_4$ is located at the circumcenter $u_{123}$ of $\langle u_1, u_2, u_3\rangle$ and the diameter of this circumcircle increases as $P$ mpvoes towards the origin. This implies that as $x$ decreases the gap ratio must increase and thus $x = \frac{0.726}{2.726}\pi$ is optimal.
\end{proof}

\subsection{Counterexample of lower bound 1.78 from \cite{ZC18}}\label{counterSec}

Theorem~2 of \cite{ZC18} claims a lower bound of $1.78$ for the gap ratio of online insertion of points onto the $2$-sphere. The result is based on Lemma~1 of \cite{ZC18} which derives the lower bound by considering three points and claiming that $1.78$ is optimal. This contradicts the (constructive) proof in the present paper of Lemma~\ref{lowerboundthm1} which shows how to generate three points whose gap ratio is $\frac{1+\sqrt{5}}{2}$ exactly\footnote{Indeed, the three points can be calculated as $u_1= (0, 1, 0), u_2 = (-0.725, -0.688, 0)$ and $u_3 = (-0.395, -0.919, 0)$, where $\zeta(u_1, u_2) = 2.3299$ and $\zeta(u_2, u_3) = \zeta(u_1, u_3) = 1.9766$.}.

The issue with Theorem~$2$ of \cite{ZC18} appears to be in the final line of reasoning of the proof. The authors derive that $0.719\pi \leq \beta \leq 0.764\pi$, where $\beta = \zeta(u_1, u_2)$ and calculate that $\gamma = \frac{\pi\beta}{2\pi-\beta}$ with $\gamma = \zeta(u_2, u_3)$. It is then claimed ``..the gap ratio is a decreasing function w.r.t $\beta$'' which appears incorrect. In order that the gap-$2$ and gap-$3$ ratios are identical, the authors derive that $\gamma = \frac{\pi\beta}{2\pi-\beta}$. Now, $\beta =  0.719\pi$ gives $\gamma = 1.763$ and thus a gap ratio of $1.78$. However this is not optimal given that $u_1, u_2, u_3, O$ lie on the same plane and thus $\zeta(u_1, u_3) = 2\pi - \beta - \gamma = 2.261 > \zeta(u_2, u_3) = \gamma$. When placing $u_3$ on the plane containing $O, u_1, u_2$, the optimal position lies equidistant to $u_1$ and $u_2$, since that maximizes the minimum distance between points, with the maximal empty circle unaffected (and of diameter $\pi$). Placing $u_3$ at this position would thus decrease the claimed lower bound of the gap ratio for three points.

\section{Conclusion}

In order to illustrate the rate of convergence of the gap ratio for various depths of tessellations starting from a single equilateral triangle of the regular icosahedron $T$, we wrote a program to perform recursive triangular dissection and to measure the minimum and maximum ratios. The results are shown in Table~\ref{resultsTab}. 

The table shows that starting from $T$, the gap ratio of the complete depth-$k$ tessellation of $T$ quickly approaches $1.38$. The next point inserted after reaching a complete depth-$k$ tessellation (with $12*4^k$ minimal triangles), requires us to split one of the edges in half according to Algorithm~\ref{thealg}. This decreases the minimum gap ratio by a factor of $\frac{1}{2}$ which increases the gap ratio by a factor of $2$. Therefore $2\cdot\rho^T(\sigma^k(T))$ shows the maximal gap ratio at any point, not only restricted to complete depth-$k$ tessellations.

\begin{table}[ht]
\footnotesize
\begin{center}
\begin{tabular}{|c|c|c|c|c|}
\hline
Depth of  & $\rho^T_{\text{min}}(\sigma^k(T))$ & $\rho^T_{\text{max}}(\sigma^k(T))$ & $\rho^T(\sigma^k(T))$ & $2\cdot\rho^T(\sigma^k(T))$ \\ Tessellation & & & &  \\
\hline
0 & 1.1071 & 1.3047 & 1.1784 & 2.3568\\
1 & 0.5536 & 0.7297 & 1.3182 & 2.6364\\
2 & 0.2768 & 0.3774 & 1.3636 & 2.7272\\
\vdots & \vdots & \vdots & \vdots & \vdots\\
7 & 0.0086 & 0.0119 & 1.3800 & 2.7600\\
\hline
\end{tabular}
\end{center}
\caption{The gap ratio of the depth-$k$ tessellation of the regular icosahedron when isolated to an equilateral spherical triangle $T$.}\label{resultsTab}
\end{table}

\begin{figure}[ht]
\begin{minipage}[c]{0.40\linewidth}
\begin{center}
          \begin{tikzpicture}[scale=1.5]
            \pgfmathsetmacro\R{1} 
            \fill[ball color=white!10, opacity=0.1] (0,0,0) circle (\R); 
            \tdplotsetmaincoords{65}{120}
            \begin{scope}[tdplot_main_coords, shift={(0,0)}]
              \coordinate (O) at (0,0,0);


\tdplotdefinepoints(0,0,0)(1, 0, 0)(0,1,0)
\tdplotdrawpolytopearc[thick, blue]{\R}{anchor=south}{$\frac{\pi}{2}$}

\tdplotdefinepoints(0,0,0)(0,1,0)(0, 0, 1)
\tdplotdrawpolytopearc[thick, blue]{\R}{anchor=west}{}

\tdplotdefinepoints(0,0,0)(0,0, 1)(1, 0, 0)
\tdplotdrawpolytopearc[thick, blue]{\R}{anchor=west}{}

\tdplotdefinepoints(0,0,0)(0,1, 0)(-1, 0, 0)
\tdplotdrawpolytopearc[dashed]{\R}{anchor=west}{}

\tdplotdefinepoints(0,0,0)(0,1, 0)(0, 0, -1)
\tdplotdrawpolytopearc[dashed]{\R}{anchor=west}{}

\tdplotdefinepoints(0,0,0)(0,0, -1)(1, 0, 0)
\tdplotdrawpolytopearc[dashed]{\R}{anchor=west}{}

\tdplotdefinepoints(0,0,0)(0,0, -1)(-1, 0, 0)
\tdplotdrawpolytopearc[dashed]{\R}{anchor=west}{}

\tdplotdefinepoints(0,0,0)(0,0, -1)(0, -1, 0)
\tdplotdrawpolytopearc[dashed]{\R}{anchor=west}{}

\tdplotdefinepoints(0,0,0)(1,0, 0)(0, -1, 0)
\tdplotdrawpolytopearc[dashed]{\R}{anchor=west}{}

\tdplotdefinepoints(0,0,0)(0,0, 1)(0, -1, 0)
\tdplotdrawpolytopearc[dashed]{\R}{anchor=west}{}

\tdplotdefinepoints(0,0,0)(-1,0, 0)(0, -1, 0)
\tdplotdrawpolytopearc[dashed]{\R}{anchor=west}{}

\tdplotdefinepoints(0,0,0)(0,0, 1)(-1, 0, 0)
\tdplotdrawpolytopearc[dashed]{\R}{anchor=west}{}

\tdplotdefinepoints(0,0,1)(0.1961,0, 0.9806)(0, 0.1961, 0.9806)
\tdplotdrawpolytopearc[thick]{0.2}{anchor=north}{$\frac{\pi}{2}$}

              \coordinate (X) at (3,0,0) ;
              \coordinate (Y) at (0,2,0) ;
              \coordinate (Z) at (0,0,2) ;

              \draw[-latex] (O) -- (X) node[anchor=west] {$X$};
              \draw[-latex] (O) -- (Y) node[anchor=west] {$Y$};
              \draw[-latex] (O) -- (Z) node[anchor=west] {$Z$};

            \end{scope}
          \end{tikzpicture}
        \end{center}
\caption{Octahedral spherical tessellation}\label{figuretessellation}
\end{minipage}
\hspace{1.5cm}
\begin{minipage}[c]{0.40\linewidth}
\begin{center}
          \begin{tikzpicture}[scale=1.5]
            \pgfmathsetmacro\R{1} 
            \fill[ball color=white!10, opacity=0.1] (0,0,0) circle (\R); 
            \tdplotsetmaincoords{65}{120}
            \begin{scope}[tdplot_main_coords, shift={(0,0)}]
              \coordinate (O) at (0,0,0);


\tdplotdefinepoints(0,0,0)(0, 0.5257, 0.8507)(0, -0.5257, 0.8507) 	
\tdplotdrawpolytopearc[thick, blue]{\R}{anchor=south}{}			
\tdplotdefinepoints(0,0,0)(0, 0.5257, 0.8507)(0.8507, 0, 0.5257)		
\tdplotdrawpolytopearc[thick, blue]{\R}{anchor=south}{}
\tdplotdefinepoints(0,0,0)(0, -0.5257, 0.8507)(0.8507, 0, 0.5257)	
\tdplotdrawpolytopearc[thick, blue]{\R}{anchor=south}{}
\tdplotdefinepoints(0,0,0)(0.8507, 0, 0.5257)(0.5357, 0.8407, 0)		
\tdplotdrawpolytopearc[thick, blue]{\R}{anchor=south}{}
\tdplotdefinepoints(0,0,0)(0, 0.5257, 0.8507)(0.5257, 0.8507, 0) 	
\tdplotdrawpolytopearc[thick, blue]{\R}{anchor=south}{}	
\tdplotdefinepoints(0,0,0)(0.8507, 0, 0.5257)(0.8507, 0, -0.5257)	
\tdplotdrawpolytopearc[thick, blue]{\R}{anchor=south}{}
\tdplotdefinepoints(0,0,0)(0.8407, 0, -0.5257)(0.5357, 0.8507, 0)	
\tdplotdrawpolytopearc[thick, blue]{\R}{anchor=south}{}
\tdplotdefinepoints(0,0,0)(0.8407, 0, -0.5257)(0.5357, -0.8407, 0)	
\tdplotdrawpolytopearc[thick, blue]{\R}{anchor=south}{}
\tdplotdefinepoints(0,0,0)(0.8607, 0, 0.5357)(0.5157, -0.8407, 0)	
\tdplotdrawpolytopearc[thick, blue]{\R}{anchor=south}{}
\tdplotdefinepoints(0,0,0)(0, -0.5257, 0.8507)(0.5257, -0.8507, 0)	
\tdplotdrawpolytopearc[thick, blue]{\R}{anchor=south}{}
\tdplotdefinepoints(0,0,0)(0, 0.5257, 0.8507)(-0.8507, 0, 0.5257) 	
\tdplotdrawpolytopearc[blue]{\R}{anchor=south}{}	
\tdplotdefinepoints(0,0,0)(0, -0.5257, 0.8507)(-0.8507, 0, 0.5257) 	
\tdplotdrawpolytopearc[blue]{\R}{anchor=south}{}	
\tdplotdefinepoints(0,0,0)(0, 0.5257, 0.8507)(-0.5257, 0.8507, 0) 	
\tdplotdrawpolytopearc[blue]{\R}{anchor=south}{}	
\tdplotdefinepoints(0,0,0)(0.5257, 0.8507, 0)(-0.5257, 0.8507, 0) 	
\tdplotdrawpolytopearc[thick, blue]{\R}{anchor=south}{}	
\tdplotdefinepoints(0,0,0)(0, 0.5257, -0.8507)(0, -0.5257, -0.8507) 	
\tdplotdrawpolytopearc[dashed]{\R}{anchor=south}{}
\tdplotdefinepoints(0,0,0)(0, 0.5257, -0.8507)(0.8507, 0, -0.5257) 	
\tdplotdrawpolytopearc[blue]{\R}{anchor=south}{}
\tdplotdefinepoints(0,0,0)(0, 0.5257, -0.8507)(-0.8507, 0, -0.5257) 	
\tdplotdrawpolytopearc[dashed]{\R}{anchor=south}{}
\tdplotdefinepoints(0,0,0)(0, 0.5257, -0.8507)(0.5257, 0.8507, 0) 	
\tdplotdrawpolytopearc[blue]{\R}{anchor=south}{}
\tdplotdefinepoints(0,0,0)(0, 0.5257, -0.8507)(0.5257, -0.8507, 0) 	
\tdplotdrawpolytopearc[dashed]{\R}{anchor=south}{}	
\tdplotdefinepoints(0,0,0)(0, -0.5257, 0.8507)(-0.5257, -0.8507, 0)	
\tdplotdrawpolytopearc[dashed]{\R}{anchor=south}{}
\tdplotdefinepoints(0,0,0)(0, -0.5257, -0.8507)(0.8507, 0, -0.5257)	
\tdplotdrawpolytopearc[dashed]{\R}{anchor=south}{}
\tdplotdefinepoints(0,0,0)(0, -0.5257, -0.8507)(-0.8507, 0, -0.5257)	
\tdplotdrawpolytopearc[dashed]{\R}{anchor=south}{}
\tdplotdefinepoints(0,0,0)(0, -0.5257, -0.8507)(0.5257, -0.8507, 0)	
\tdplotdrawpolytopearc[dashed]{\R}{anchor=south}{}
\tdplotdefinepoints(0,0,0)(0, -0.5257, -0.8507)(-0.5257, -0.8507, 0)	
\tdplotdrawpolytopearc[dashed]{\R}{anchor=south}{}
\tdplotdefinepoints(0,0,0)(-0.8507, 0, 0.5257)(-0.8507, 0, -0.5257)	
\tdplotdrawpolytopearc[dashed]{\R}{anchor=south}{}
\tdplotdefinepoints(0,0,0)(-0.8607, 0, 0.5157)(-0.5457, 0.8307, 0)	
\tdplotdrawpolytopearc[dashed]{\R}{anchor=south}{}
\tdplotdefinepoints(0,0,0)(-0.8307, 0, 0.5357)(-0.5157, -0.8507, 0)	
\tdplotdrawpolytopearc[dashed]{\R}{anchor=south}{}

\draw[] (0, 0.5257, 0.8507) node[anchor=north west] {$u_1$};
\draw[] (0, -0.5257, 0.8507) node[anchor=south east] {$u_2$};
\draw[] (0, 0.5257, -0.8507) node[anchor=north] {$u_3$};
\draw[] (0.8507, 0, 0.5257) node[anchor=west] {$u_5$};
\draw[] (0.8507, 0, -0.5257) node[anchor=north east] {$u_6$};
\draw[] (-0.8507, 0, 0.5257) node[anchor=south] {$u_7$};
\draw[] (0.5257, 0.8507, 0) node[anchor=north west] {$u_9$};
\draw[] (0.5257, -0.8507, 0) node[anchor=east] {$u_{10}$};
\draw[] (-0.5257, 0.8507, 0) node[anchor=west] {$u_{11}$};

              \coordinate (X) at (3,0,0) ;
              \coordinate (Y) at (0,2,0) ;
              \coordinate (Z) at (0,0,2) ;

              \draw[-latex] (O) -- (X) node[anchor=west] {$X$};
              \draw[-latex] (O) -- (Y) node[anchor=west] {$Y$};
              \draw[-latex] (O) -- (Z) node[anchor=west] {$Z$};

            \end{scope}
          \end{tikzpicture}
        \end{center}
\caption{Icosahedral spherical tessellation}\label{figureisotessellation}
\end{minipage}
\end{figure}

To evaluate the most appropriate initial shape for our algorithm, we derived (both theoretically and with a computational simulation) the gap ratios of the stage 1 and stage 2 tessellations of various Platonic solids, shown in Table~\ref{results2Tab}. The results for the dodecahedron are from \cite{walcom17} using a different tessellation (the dodecahedron has non triangular faces).

\begin{table}[ht]
\footnotesize
\begin{center}
\begin{tabular}{|c|c|c|c|c|}
\hline
 & Tetrahedron & Octahedron & Dodecahedron & Icosahedron  \\
\hline
Stage 1 & 2.289 & 2.0 & 2.618 & \emph{2.8376} \\
\hline
Stage 2 & \emph{5.921} & \emph{3.601} & \emph{5.995} & 2.760 \\
\hline
\end{tabular}
\end{center}
\caption{The gap ratio of stage one and two of various regular Platonic solids. Italic elements show which value defines the overall gap ratio in each case.}\label{results2Tab}
\end{table}

 The results match our intuition, that a finer grained initial tessellation such as that from an icosahedron performs much better in stage 2 than a more coarse grained initial tessellation such as that from a tetrahedron. This is illustrated by Fig.~\ref{skewedtess} which shows that for a large initial equilateral spherical triangle, the recursive triangular dissection procedure deforms the four triangle by a larger margin. The regular icosahedron thus has the essential criteria that we require; it has a low stage 1 and stage 2 gap ratio, and it is a regular tessellation into equilateral spherical triangles. It would be interesting to consider modifications of the stage 2 procedure which may allow the octahedron to be utilised, given its low stage 1 gap ratio. This may require some a modification of Lemma~\ref{mainlem} which works also with non equilateral spherical triangles. 

It would also be beneficial to tighten the lower bound. In Theorem~\ref{lowerboundthm1} we derive the lower bound of $1.618$ by considering the first three points and in Theorem~\ref{lowerboundthm2} we consider the first four points to derive a lower bound of $1.726$. Our proof technique shows that considering the first three or four points gives one or two degrees of freedom respectively. Even with two degrees of freedom the reasoning becomes significantly more challenging and therefore a different approach may be necessary in order to derive a reasonable increase in the lower bound.

\end{document}